\documentclass[10pt, conference]{IEEEtran}

\IEEEoverridecommandlockouts

\usepackage{amsmath}
\usepackage{verbatim}
\usepackage{color}
\usepackage{algorithm}
\usepackage{xspace}
\usepackage{times}
\usepackage{caption}
\usepackage{subcaption}

\csname @openrightfalse\endcsname

\usepackage{algorithmic}

\newcommand {\ie} {{\em i.e., }}
\newcommand {\eg} {{\em e.g., }}

\newcommand {\beq} {\begin{equation}}
\newcommand {\eeq} {\end{equation}}
\newcommand {\bequn} {\begin{equation*}}
\newcommand {\eequn} {\end{equation*}}
\newcommand {\bear} {\begin{eqnarray}}
\newcommand {\eear} {\end{eqnarray}}

\DeclareMathOperator{\Var}{Var}
\newcommand {\Eqref}[1]{Eq.~(\ref{#1})}

\newcommand{\fig}[1]{Figure~\ref{#1}}
\newcommand{\figs}[2]{Figures~\ref{#1} and~\ref{#2}}
\newcommand{\figsm}[2]{Figures~\ref{#1} to~\ref{#2}}

\def\Sys{Sunstar\xspace}  

\usepackage{paralist} 
\let\itemize\compactitem 
\let\enditemize\endcompactitem 
\let\enumerate\compactenum 
\let\endenumerate\endcompactenum

\title{\Sys: A Cost-effective Multi-Server Solution\\ for Reliable Video Delivery}
\author{Behnaz Arzani*, Nicholas Iodice*, Steven Hwang*, Prahalad Venkataramanan*, Roch Geurin$\dagger$, Boon Thau Loo*\\
*University of Pennsylvania, $\dagger$ University of Washington at St. Louis}

\newcommand{\Section}[1]{\vspace{-0.75mm}\section{#1}\vspace{-0.75mm}}

\newcommand{\Subsection}[1]{\vspace{-0.75mm}\subsection{#1}\vspace{-0.75mm}}

\newcommand{\eat}[1]{}

\eat{
\numberofauthors{3}
\author{
\alignauthor
       Behnaz Arzani \\
       \affaddr{University of Pennsylvania}\\
       \email{barzani@seas.upenn.edu}
\alignauthor
       Roch Guerin \\
       \affaddr{WUSTL}\\
       \email{guerin@wustl.edu}
\alignauthor Boon Thau Loo \\
       \affaddr{University of Pennsylvania}\\
       \email{boonloo@cis.upenn.edu}
}

}

\eat{
\author[1]{Behnaz Arzani\thanks{barzani@seas.upenn.edu}}
\author[2]{Roch Guerin\thanks{guerin@wustl.edu}}
\author[1]{Boon Thau Loo\thanks{boonloo@cis.upenn.edu}}
\affil[1]{University Of Pennsylvania}
\affil[2]{University of Washington at St. Louis}

}

\usepackage[svgnames]{xcolor}
\newcommand{\TODO}[1]{{\color{Navy}((#1))}}

\usepackage[section]{placeins}
\usepackage{url}
\usepackage{amsmath}
\usepackage{tikz}
\usepackage{subcaption}
\usepackage{graphicx}
\graphicspath{{images/}}

\usepackage{amsthm}
\theoremstyle{plain}
\newtheorem{theorem}{Theorem}

\begin{document}

\maketitle



\begin{abstract}
In spite of much progress and many advances, cost-effective, high-quality video delivery over the Internet remains elusive. To address
this ongoing challenge, we propose \Sys, a solution that leverages
simultaneous downloads from multiple servers to preserve video
quality. The novelty in \Sys is not so much in its use of multiple
servers but in the design of a scheduler that balances improvements
in video quality against increases in (peering) costs.  The paper's
main contributions are in elucidating the impact on cost of various
approaches towards improving video quality, including the use of
multiple servers, and incorporating this understanding into the design
of a scheduler capable of realizing an efficient trade-off.  Our results show that 
\Sys's scheduling algorithm can significantly improve performance (up to $50\%$ in some instances) without
cost impacts.

\eat{
In spite of much progress, consistent high-quality video delivery over
the Internet remains elusive. The paper tackles this challenge through
an application layer multipath solution that lets video clients
simultaneously download from multiple servers.  Clients distribute
requests for video chunks across servers based on the solution of a
linear optimization model that seeks to minimize rate variations while
meeting minimum rate constraints.  The paper demonstrates the efficacy
of the approach but also highlights that its benefits are often at
the expense of significantly higher peering costs for the video
provider.  The paper then develops a cost-aware server selection
algorithm that preserves the video quality improvements a multipath
solution affords while mitigating its impact on peering costs.  An
implementation of the video client and server selection algorithm are
used to experimentally validate the results on the Emulab testbed.
}

\end{abstract}

\Section{Introduction}
\label{sec:intro}

The growing importance of video traffic is by now well-documented,
with the share of video traffic on North American networks exceeding
$70\%$ of peak hour traffic in early 2016 and expected to surpass
$80\%$ by the end of 2020~\cite{sandvine16}. And while its dominance
is not as strong on mobile access links where it now represents about
$40\%$ of peak traffic, it is a major factor there as well.
Furthermore, this growth appears unimpeded by continued progress in
codec efficiency, \eg x265 or VP9
codecs~\cite{kufa15,uhrina14,ozer16}.

In the face of such a trend, or maybe because of it, there are,
however, persistent problems when it comes to ensuring quality video
delivery.  For example, Conviva VXR reports,
\eg\cite{vxr14,vxr15,conviva15}, which each year track various video
performance metrics, report that buffering events (periods during
which the video stalls to replenish its playback buffer) remain common
(from $20\%$ to $40\%$ depending on location), as do drop in video
resolution while playing (those affect over $50\%$ of all
videos). Equally important, degradation in video quality is also known
to have a major impact on users' behavior and their satisfaction with
Internet video~\cite{krishnan12}.  This is obviously of concern, as
articulated in several recent industry forums focused on video
delivery~\cite{cdnsummit13,cdnsummit14,cdnsummit15,cdnsummit16}.

The first goal of this paper is, therefore, to explore possible
solutions to remedy this situation and improve the quality of Internet
video delivery.  Of interest in this context is the use of multipath
solutions, and in particular solutions that let video clients download
video segments (chunks) simultaneously from multiple servers.
We refer to such solutions as {\em Multipath-MultiServer (MuMS)},
\eg\cite{dynamic,youtuber}.
Reliance on paths from multiple distinct servers can help mitigate
exposure to quality degradations due to congestion or failure
of individual paths, and server overload.  In particular, multipath
has been shown useful in improving throughput and reliability in
both wired and wireless
networks~\cite{apostolopoulos04,Chen:2004ih,chen09,ganesan2001highly,golubchik02,Radi:2012en},
and there is initial evidence that it could also benefit
delay~\cite{javed09} as well as rate stability~\cite{arzani12}.  The
latter is of particular interest, as rate variations are a major
factor in video quality degradation.

Specifically, a typical video streaming client operates in two
phases~\cite{rao11}: pre-buffering and re-buffering.  A video is split
into segments (chunks), and in the pre-buffering phase the client
requests chunks at the maximum rate allowed by the
network\footnote{This holds for both video download and
  live-streaming, though with obvious limitations on the pre-buffering
  phase for the latter.}. Once there are enough chunks in the client's
playback buffer, it starts playing the video and switches to
re-buffering mode. In this mode, the client requests chunks at a fixed
rate determined by the encoding rate.  Video playback proceeds
smoothly as long as chunks are in the playback buffer before they need
to be played-out.  Variations in transmission capacity between the
server(s) and the client can result in late delivery or loss of video
chunks.  This, in turn, depletes the playback buffer and eventually
induces video stalls or skips.

A popular approach for dealing with variations in transmission
capacity is adaptive bit rates (ABR)~\cite{abr96,abr16} -- sometimes
also called HTTP Adaptive Streaming (HAS).  Under ABR/HAS, the client
changes its encoding bit rate to match the capacity available in the
network and avoid losses/delays.  However, rate changes remain visible
to the viewer and still translate in degraded quality of experience
(QoE)~\cite{sogaard16,garcia14,seufert14}, albeit at a lesser
level. As a result, it remains desirable to devise solutions that
eliminate or mitigate variations in transmission rate, and 
preserve video quality without having to resort to (coding) rate
adjustments.  This is one of the goals of our solution, \Sys, which
seeks to leverage multiple paths to different servers to maintain a
stable transmission rate even in the presence of network variations
(on individual paths).  In this respect, \Sys is complementary to
ABR-based solutions.


Another goal of \Sys is to realize its goal of mitigating rate
variations without impacting peering costs, \ie the costs video
providers incur from Internet Service Providers (ISPs) for delivering
video to their customers.  Ensuring that better quality in video
delivery does not translate into higher peering costs is of particular importance given the low profit margins under which most video
providers operate~\cite{ball15,bonte10}.  Of concern in our
context is the extent to which reliance on multiple paths might
increase peering costs.  The possibility of such increases is, in
hindsight, intuitive given the non-linear nature of most ISPs'
charging model, \ie
most charge based on the $95^{\mbox{th}}$ percentile of usage in
$5$~mins intervals over a period of a month\footnote{See
  \url{https://www.noction.com/blog/95th_percentile_explained.}}.
Spreading transmissions over multiple paths means that a separate
$95^{\mbox{th}}$ percentile is now computed on each path, which, as
discussed in Section~\ref{sec:mums}, can result in a higher overall
cost value.

In summary, \Sys's main contributions are as follows:

\noindent{\bf A cost-neutral solution to improving video delivery.}
We develop a principled understanding of how different mechanisms
for improving video delivery, including multipaths, contribute to
higher (peering) costs.  We use this understanding to develop a
scheduler capable of delivering significant performance improvements
with little to no impact on cost.

\noindent{\bf A video client improving users' QoE.}  We implement a
\Sys video client in user space, and demonstrate its benefits by
quantifying its ability to minimize video stalls/skips across a broad
range of network impairments.  Our Emulab~\cite{emulab} results
indicate that \Sys can improve video performance quality by up to
$50\%$ or more in various settings.

\Section{Related Works}

There has obviously been much work on optimizing video transmission
and using multipaths to overcome network impairments.  Our intent is
not to provide an exhaustive review, but rather to summarize major
approaches and highlight similarities and differences with this work.

\Subsection{Video Delivery Optimizations}

\noindent {\bf Adaptive bit rate (ABR)}~\cite{abr96,abr16} is, as
mentioned earlier, a powerful approach for mitigating the impact of
network rate variations by allowing clients to correspondingly adjust
their video coding rate.  The main drawback of ABR is that it requires
servers and caches to store multiple encodings of the same video, or
codecs to be able to dynamically update their coding rate.  In
addition, adjustments in coding rates still produce noticeable changes
in video quality~\cite{sogaard16,garcia14,seufert14}. Our goal with
\Sys is complementary to ABR, in that we aim to leverage multiple
paths to different servers to minimize (network) rate variations, and
therefore coding rate changes.

\noindent {\bf Caching} is another popular strategy.  It improves
performance by moving files as close as possible to clients through
caches located at the network edge.  This is, however, not always
effective, in part because copyrights laws make much content
un-cachable, and the combination of the long tail of video
popularity~\cite{longtail} and the use of ABR can lower cache
efficiency.  Consequently, even smart caching algorithms
only boast a cache efficiency of about
$50\%$~\cite{qwiltreport} (for ABR videos). \Sys is meant to improve
video delivery in instances when video cannot be served from a local
edge cache.

\noindent {\bf OpenConnect}~\cite{openconnect} was proposed by
Netflix. It relies on embedding appliances in ISPs' networks to locate
content closer to clients and to preemptively populate caches at
off-peak hours to avoid cache warm-up and network congestion during
peak hours. It calls for partnership between content providers and
ISPs, \eg locating appliances in the ISP's facility, which some large
ISPs are reluctant to engage in as they have competitive
businesses~\cite{netflixdeal}.
 
\noindent {\bf Content filtering} limits content available to
users to content that can be delivered with high quality, \eg
from caches.  This is realized by applying filters that limit viewable
listings to a subset of (popular) videos, or by steering users away
from unpopular items~\cite{UCG}. Both approaches result in potential
loss of revenue, \eg removing such filters can increase the number of
views by $45\%$~\cite{UCG}.

\noindent {\bf Dynamic CDN switching} is offered by companies such as
Conviva and Cedexis which act as brokers in the CDN domain.  They
measure CDN status and switch between CDN 
providers based on performance. The main disadvantage is that CDNs no
longer control how clients are redirected to servers, which can have
unintended consequences, including higher costs~\cite{broker}. 
In contrast, \Sys keeps the assignment of clients to servers under the
video provider's purview and incorporates mitigating cost increase as
an explicit criterion.

\noindent{\bf Hybrid CDN-P2P} seeks to combine the best of CDN and
peer-to-peer solutions~\cite{cdnp2p}.  Netsession~\cite{netsession}
offers a representative example of the potential benefits of such an
approach.  Unlike \Sys, it again does not offer an explicit control on
how improving performance affects a provider's cost.  Additionally,
aspects such as copyright management are traditionally difficult to
handle in a P2P setting.

\eat{

\Subsection{Background: ISP and Video Delivery}
\label{sec:background}

\noindent{\bf{Peering Cost: }}We start by providing some background.  Traditionally, ISPs charge
CDNs based on a {\em percentile charging model}. The ISP monitors the
CDN's traffic volume at the peering link and records its value every 5
minutes. At the end of a billing cycle, traffic records are sorted
and, in a \textit{$q$-percentile} model, the upper $(1-q)\%$ are
discarded. The largest record remaining is the {\em charging volume}
that determines the CDN's cost.  Given the dependency of peering costs
on traffic temporal variations, MuMS clients shifting traffic from one
link to another (to improve their performance) can clearly affect
costs.

\noindent{\bf{Clients: }}The operation of a typical video streaming client has two
phases~\cite{prebuf}: pre-buffering and re-buffering.  A video is
split into frames (chunks), and in the pre-buffering phase, the client
requests chunks at the maximum rate allowed by the network. Once
enough chunks are accumulated in the buffer, the client switches to
the re-buffering phase, where it requests chunks at a fixed rate
determined by the encoding rate. For a single-path client, the
encoding rate fully determines the number of chunks that need to be
requested at any point in time.

MuMS clients are different, as they split their
requests between servers they have been assigned. As single path
clients, a MuMS client fetches chunks at the maximum possible rate
from all its servers during the pre-buffering phase. In the
re-buffering phase, the aggregate bandwidth available to a MuMS client
is typically higher than what it needs.  Its goal is
then to distribute requests across servers to meet rate guarantees
while maximizing rate stability.  This calls for estimating the
available bandwidth to each server and distributing requests to avoid
congesting the network and/or servers. This is the task assigned to the MuMS 
scheduler.
}
\eat{The vast research on transport layer protocols proves useful here, as
it can be re-purposed in an application layer implementation. Namely,
the MuMS client can use a method similar to that of TCP to estimate
the bandwidth delay product to each server\footnote{The idea of using
  a TCP-like algorithm is not new, \eg\cite{multihoming2}.}. The
resulting per server windows provide upper bounds on the number of
requests the client may have outstanding to each server without
incurring the risk of congestion. The MuMS scheduler can then decide
how many requests to send to each server to meet its target rate and
maximize rate stability.  Those decisions, as we shall see, can affect
both performance and cost.}


\eat{
The extent of this impact will depend on the client's behavior, and
Section~\ref{sec:highlevel} offers a brief overview
of a typical MuMS client design, and therefore the type of behaviors
it can give rise to.  The potential impact on cost of those behaviors
is then illustrated in Section~\ref{sec:tradeoff} through two
representative scenarios.  Section~\ref{why} offers a more systematic
investigation of when and how a MuMS client optimizing its performance
can affect cost.  In so doing, it develops an understanding the
relationship between performance and cost, which is leveraged in
Section~\ref{sec:design} to design a MuMS system capable of realizing
an efficient trade-off.  The design's ability to meet its goals are
evaluated in Section~\ref{sec:evalmums}.}

\Subsection{Multipath Solutions}

The benefits of multipaths have been studied in
numerous settings, \eg\cite{conga,MRTP,mmsys}, but perhaps most
visible among them are studies of Multipath TCP
(MPTCP)~\cite{nsdi2012}, whose investigations related to congestion
control~\cite{conj1,conj2,analdesign,arzani14} or
scheduling~\cite{pams14,mptcpscheduler14}
are of most relevance, even if not directly applicable because of MPTCP's
assumption of a single source and a single destination. Nevertheless,
several techniques developed to improve MPTCP can be repurposed in a
MuMS setting.  For example, as discussed in Section~\ref{sec:design},
\Sys is able to leverage MPTCP's opportunistic retransmit to improve
its performance.

More directly comparable to \Sys are works that explicitly target
improving video delivery by relying on multiple servers. In most such
settings, \eg\cite{dynamic,youtuber}, the focus has, however, been on
optimizing download \emph{rates}, which, as we shall see, can have a
significant impact on cost.  Specifically and as discussed in
Section~\ref{sec:mums}, while the more aggressive download strategies
of~\cite{dynamic,youtuber} can reduce the odds of skips and stalls,
they typically result in higher costs.  In contrast, \Sys aims to
realize comparable improvements in video quality, but with little to
no increases in cost.


\Subsection{Server selection and cost optimizations}

Another relevant body of work is that of server selection algorithms
that optimize for a given metric, \eg performance or
cost~\cite{centralized,donar,donar2,ishai}. Extending those approaches
to a MuMS' setting is challenging as the
multipath nature of MuMS clients makes predicting variations in
traffic volumes at peering links more difficult than with single path
clients.  In particular, clients are now free to choose how to
distribute video requests across paths.  How this impacts cost and
performance adds a new non-trivial dimension to the problem.  

In this context, the approach closest to \Sys is~\cite{multihoming2}.
It considers both performance and cost and adopts a cost
minimization formulation with performance as a constraint, where for
each CDN performance is based on long-term QoE measurements from
clients in different regions.  Given an expected request load,
\cite{multihoming2} computes a ``prioritized'' list of servers that a
client should use when requesting content.  Higher
priority servers are to be used first as long as they have available
capacity.  A TCP-like AIMD mechanism is used to estimate the bandwidth
available to each server.  \Sys's approach differs from that
of~\cite{multihoming2} in that rather than minimizing cost and keeping
performance as a constraint, it leverages its understanding of the
relationship between cost and performance to select rate variation as
its minimization target.  In addition, \Sys's scheduler offers a more
responsive mechanism than that of~\cite{multihoming2}, which is
limited to the set of servers computed by its optimization.  In some
sense, the scheduler of~\cite{multihoming2} is similar to the min-RTT
scheduler of Section~\ref{sec:eval}, which, as we shall see, performs
significantly worse than that of \Sys in terms of both cost and
performance.

Finally, of note in the context of cost optimization
is~\cite{shapley}, which attempts to account for the contributions of
individual users to a $95^{\mbox{th}}$~percentile cost function. 
Although such an approach could be used to formulate an appropriate objective function for a server selection algorithm, it requires
detailed knowledge of the exact traffic patterns of each user.  This
is unlikely to be feasible, especially in a multipath setting where
variations on a given path affect traffic on all paths.

\Section{MuMS Benefits and Implications}
\label{sec:mums}


Before presenting \Sys, we first motivate a MuMS approach by showing
the type of performance improvements achievable when downloading from
multiple servers. Next, we offer insight into the relationship that
exists between performance and (peering) cost, and in particular why
reliance on multiple paths, as in MuMS, can result in higher costs.
This illustrates the need for an approach that balances performance
and cost, and offers a possible direction for \Sys's design towards
realizing such a goal.

\Subsection{Performance Benefits}

A MuMS solution should improve users' QoE as multiple servers (and the
paths from those servers) are unlikely to simultaneously experience
congestion or failures.  To assess the significance of those
gains, we compare the performance of a MuMS client to that of a
single-server client in an Emulab experiment.  To simplify our setup,
in all cases clients use only a single path to each server. 

The Emulab connection between the client and each server consists of
two links separated by a shaping node with a buffer size of
50~packets.  To create an environment that exercises bandwidth
limitations, the available bandwidth to each server from the client
was set to an average value equal to its download rate $T$ (as
determined by the video player), but with variations between
$\frac{T}{2}$ and $\frac{3T}{2}$. This was achieved using
dummynet~\cite{dummynet} on the shaping node.  Each client repeatedly
downloads a large video ($5$ minutes or more) and is assigned to a
fixed set of $n$ servers where $n$ is either $1$, $2$, or $3$.  In all
cases, video download proceeds by issuing {\tt http get} requests at a
rate commensurate with that of the video. In the single-server case,
TCP controls the actual download rate.  In the multi-server case, a
standard TCP-like application-level congestion control mechanism
determines the available rate from each server, and when multiple
servers are available, the lowest RTT server is selected.  This
represents a relatively basic ``scheduler,''
which nevertheless serves the purpose of demonstrating the
benefits of a MuMS solution. 

In evaluating performance, we focus on two metrics of importance to
video QoE, namely, stalls and skips.  Clients stall whenever their
playback buffer runs empty or the next chunk to playback is
missing. Skips occur, usually after a stall, when the player decides
to skip a missing chunk and resumes playback using a latter chunk.
Previous studies~\cite{QoE1,QoE2,QoE3} have verified the correlation
between user satisfaction and these metrics. We further verify their
impact on user experience through a Mechanical Turk experiment
described in Appendix~\ref{appendix:turk}. Note that there is an
inherent trade-off between the two metrics. A short time-out for
chunks increases skips but minimizes stalls and vice versa.


\begin{figure}[hbt]
\centering
 \includegraphics[width=0.8\linewidth]{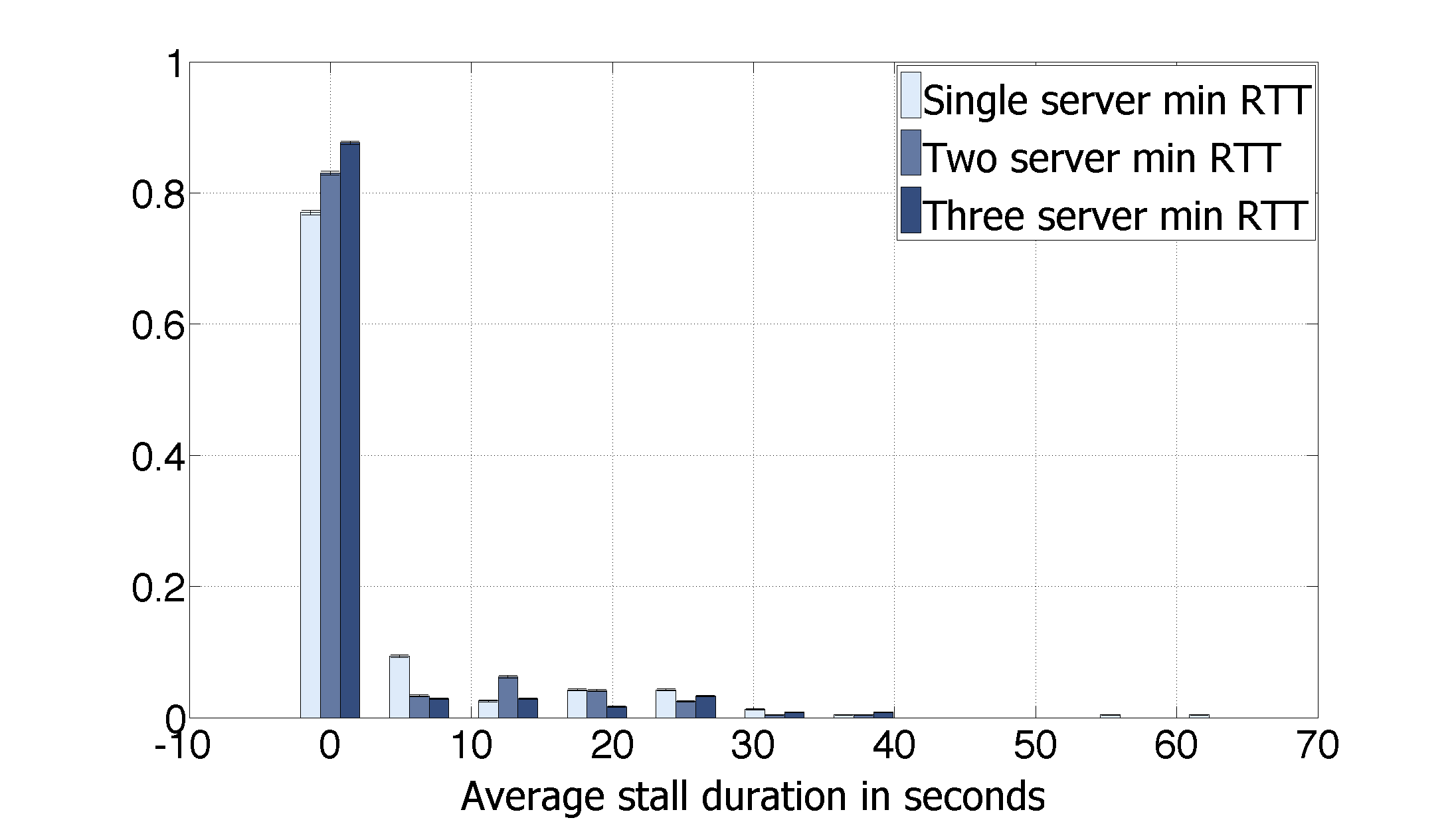}
\vspace{-1mm}
     \caption{Average stall duration across clients. \label{fig:stallsSinglePathMPTCP}}
\vspace{-1mm}
\end{figure}

\begin{figure}[hbt]
\centering
 \includegraphics[width=0.8\linewidth]{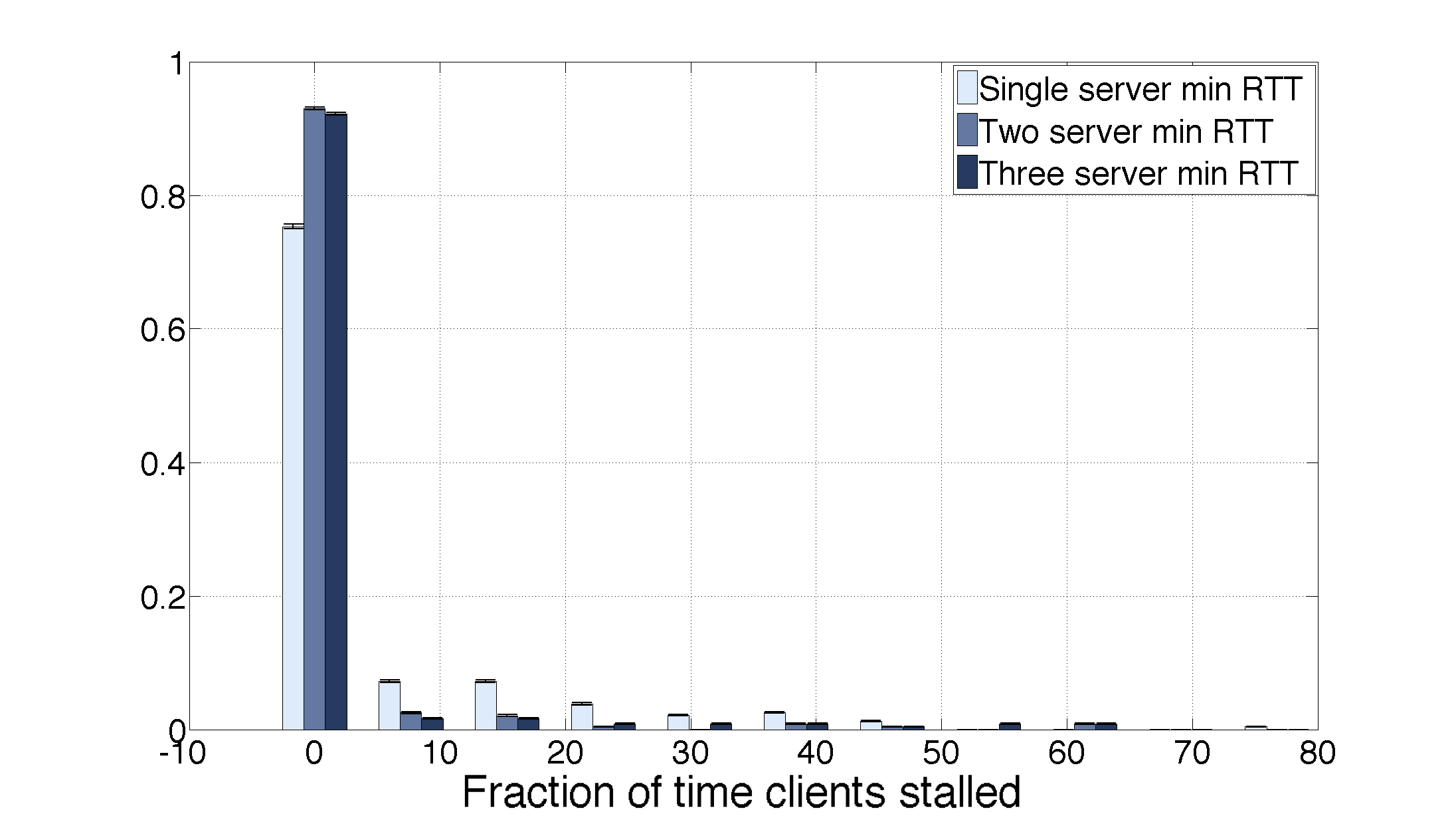}
\vspace{-1mm}
     \caption{Fraction of total download time each client was stalled (across clients). \label{fig:stallsSinglePathMPTCPfreq}}
\vspace{-1mm}
\end{figure}
\figs{fig:stallsSinglePathMPTCP}{fig:stallsSinglePathMPTCPfreq}
report the distributions (across clients) of the average stall time
and the average fraction of time clients are stalled for
configurations with $1,2$, and $3$ servers, respectively.
\fig{fig:skippedSinglePathMPTCP} focuses on the distribution of the
number of skipped video chunks for the same configurations. The
confidence intervals in the figure show the $95$ percent
confidence interval assuming a binomial distribution for data
in each bin. The figures establish the benefits of a MuMS solution,
which reduces \emph{both} stalls duration and the
number of skips.  For example, more than $91\%$ of $3$-server clients
saw no stalls, while the number was $90\%$ for $2$-server clients, and
$73\%$ for single-server clients; and those benefits persist among
clients that experienced longer stalls on average.  Similarly, over
$90\%$ $(82\%)$ of $3$-server ($2$-server) clients did not experience
any skips, while this number drops to below $80\%$ for single-server
clients, with again the benefits of a MuMS solution extending to the
tail of the distribution.

\begin{figure}[hbt]
  \centering
\includegraphics[width=0.75\linewidth]{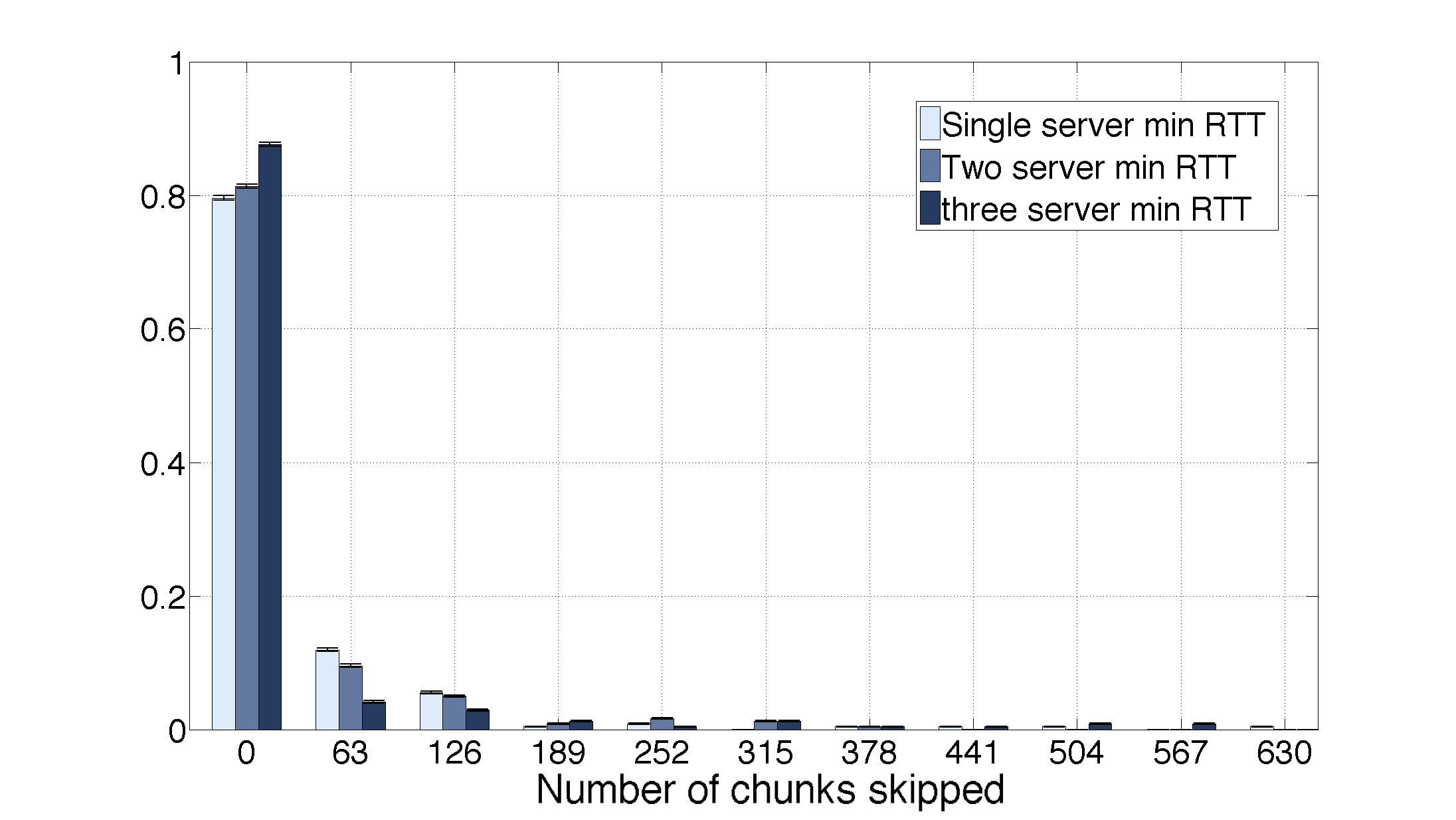}
\vspace{-1mm}
     \caption{Distribution of the number of skipped
       chunks. \label{fig:skippedSinglePathMPTCP}} 
\vspace{-1mm}
\end{figure}
In Section~\ref{sec:eval}, we show how \Sys's more sophisticated
scheduler can yield even further improvements.

\Subsection{Cost-Performance Trade-offs}
\label{sec:trade-off}
There are multiple options to improve video delivery.  A simple
approach is to download the video at the highest possible rate from
one or more servers, as it should minimize the odds of a chunk
arriving late and/or the playback buffer running empty\footnote{Note
  though that this comes at the expense of larger playback buffers at
  the clients, and a potential waste of bandwidth when clients abandon
  watching the video halfway.}.  This is why works such
as~\cite{youtuber} focus on maximizing clients' throughput during
their ``on'' period (the time during which the client fills its
playback buffer).

However, while performance benefits are intuitive, it is unclear how
such a scheme affects \emph{cost}. On one hand, a strategy
that maximizes throughput has clients leaving the system earlier,
which can reduce bandwidth usage when computed over $5$ minute
intervals.  On the other hand, the higher download rates while clients
are present can increase bandwidth usage.  How these two opposing
factors contribute to $95^{\mbox{th}}$-percentile costs is not obvious
at first sight.

To gain a better understanding of this trade-off, we develop a simple
analytical model to evaluate the impact of the download rate on
peering costs.  We consider a scenario where: (1) clients connect to a
fixed set of $k$ servers with distinct peering links for each server;
(2) clients download a video of size~$S$ at a constant aggregate rate
of $T$ and distribute download requests equally across servers (the
download rate for each server is $T/k$);
(3) clients arrive according to a Poisson process of rate $\lambda$;
and (4) peering costs follow a $q$-percentile model ($q=95$ in a
typical scenario). 

Theorem~\ref{theo:cost} establishes that the higher the download
rate~$T$, the higher the peering cost.  In other words, while a more
aggressive download strategy may improve performance, it results in
higher costs.

\begin{theorem}
\label{theo:cost}
Given clients arriving according to a Poisson process and downloading
equally from $k$~servers at an aggregate rate of~$T$, the
$q$-percentile peering cost is an increasing function of~$T$.
\end{theorem}



\begin{proof}
Assuming a properly provisioned system, \ie non-blocking, each server
plus peering link combination behaves as an $M/G/\infty$ system whose
occupancy probability is given by $\pi_i=\frac{e^{-\rho}\rho^i}{i!}$,
where $\rho=\frac{\lambda S/k}{T/k}=\frac{\lambda S}{T}$.  Assuming
$\rho$ is large, \ie we are dealing with large systems, $\pi_i$ can be
approximated by a normal distribution with mean and variance equal to
$\rho$.

The $q$-percentile occupancy $n(q)$ of the system (each client is
assigned one ``server'' with a service/download rate of $T/k$) is
then of the form  
\begin{align*}
\phi \left(\frac{n(q)-\rho}{\sqrt{\rho}}\right)=q
\end{align*}
where $\phi(x)$ is the CDF of the normal distribution. This implies
$\frac{n(q)-\rho}{\sqrt{\rho}}=\alpha$ where $\alpha > 0$ is constant, \eg for $q=0.95$, $\alpha=1.64$.  Thus,
$n(q)=\alpha\sqrt{\rho}+\rho$, and the $q$-percentile traffic volume
on the corresponding peering link is $n(q)T/k$.  Hence,
under a $q$-percentile cost model,
the peering cost for the system is $C_q(\lambda,T,k)\sim
(c\sqrt{\rho}+\rho)T/k$, an increasing function of $T$. 
\end{proof}

Although Theorem~\ref{theo:cost} relies on a number of simplifying
assumptions, it nevertheless captures the key factor that while
increasing download rates allows clients to leave the system faster
its overall impact on cost is negative.  In other words, downloading
at the lowest possible rate that meets the video requirements yields
the lowest cost.  As described in Section~\ref{sec:design}, we
leverage this insight in designing the \Sys scheduler.

\Subsection{Impact of MuMS on Cost}
\label{sec:why_cost}

The previous section established that a greedy/aggressive download
strategy had a negative impact on cost.  In this section, we show that
the multiple paths of a MuMS' solution have a similar effect.  For
that purpose and without loss of generality, we consider a system with
a single server reachable by clients over either one or two peering
links.  When two peering links are available, clients split their
traffic across the two links according to some strategy. Under the two
peering links configuration, we denote as $\pmb{x}_1$ and $\pmb{x}_2$
the vectors of traffic volumes recorded in $5$~minute intervals on
links~$1$-$2$, respectively.  Correspondingly, traffic on the peering
link of the single server configuration is denoted as~$\pmb{x}$.
Ignoring the impact of short time-scale traffic fluctuations, we have
$\pmb{x}=\pmb{x}_1+\pmb{x}_2$, and therefore
$\max\{\pmb{x}\}=\max\{\pmb{x}_1+\pmb{x}_2\}\leq\max\{\pmb{x}_1\}
+\max\{\pmb{x}_2\}$.
Hence, under a peak rate charging model, a single path solution yields
a lower cost.


There are obviously additional factors at play when considering the
multiple servers of MuMS clients and a $95^{\mbox{th}}$ percentile
rather than peak charging model.  Nevertheless, this captures a
fundamental aspect of multipath solutions and their 
impact on many non-linear cost functions.  Mitigating this impact
is one of the goals of \Sys.


\eat{
\Section{Related Work}

\noindent \textbf{MuSS:} 

\noindent \textbf{MuMS:} 

\noindent\textbf{Studies on QoE} Significant studies have been conducted to identify what video performance
metrics have the most significant impact on client performance e.g.~\cite{QoE1,QoE2}. While these work guided us in our designs 
they are orthogonal to the goals of this paper, which is to show that it is possible to achieve significantly higher QoE through multi-server clients.

}

\eat{ 
\textbf{Measurement studies:} Finally, prior measurement studies have
been conducted to dissect server selection in popular video delivery
services such as YouTube~\cite{youtube}, Netflix~\cite{netflix}, and
Hulu~\cite{hulu}. These work provide further evidence that content
providers do not assign servers to clients solely based on
performance, but rather cost and load balancing play an important part
in the decision process.
}

\section{SunStar Client}
\label{sec:design}





\fig{fig:architecture} shows the overall architecture of the \Sys
client from the point of view of a single client. \Sys runs as a
middleware between the video player (and its codec) and the network
layer in the operating system.
The \Sys middleware initiates a series of {\tt http} connections to
remote servers which store the video content. These remote servers are
themselves part of larger server farms where the same video content is
replicated on multiple servers. They are
assigned to a client requesting a video through manifest files that
the client downloads from an initial designated server, and are chosen
through server selection algorithms imposed by the provider. Multiple
servers within the same server farm may be selected.

The internals of the \Sys middleware are as follows. A multi-threaded
{\em {\tt http} connection pool} maintains connections to the selected
servers. For a given video, the \Sys middleware determines the next
set of video chunks to download through each of the available {\tt
  http} connections in the connection pool.  This allocation is
determined by the {\em Scheduler} (\S~\ref{sec:scheduler}). We revisit
the scheduler formulation later in the section, but in a nutshell its
goal is to allocate requests across servers so as to realize the best
possible trade-off between cost and performance (video quality). Based
on the results of \S\ref{sec:trade-off} this boils down to reducing
the download rate as much as possible to control cost, while ensuring
that client's performance does not suffer\footnote{Note, that as long
  as video quality remains high, clients have an incentive to
  cooperate with any video-download strategy.}.

Note that barring an unacceptably long pre-buffering phase (to
build-up a large playback buffer), blindly downloading at a constant
low rate (as assumed in the calculations of~\S\ref{sec:trade-off})
will result in poor performance, as it makes the playback buffer
vulnerable to lulls in network bandwidth. Such lulls are unavoidable
in the Internet, and the challenge is then in compensating for them
with rate increases that are as low as possible while still avoiding
instances of buffer underflow. 
\Sys's scheduler relies on a mathematical model for quantifying this
trade-off, and uses it to determine how to adapt its download rate
from each server. More formally, the scheduler aims to achieve an
(average) target rate of $T$ while ensuring minimum variations around
it. We peg $T$ to be at least the playback rate of the client
to make sure that the playback buffer is never empty, but
choose it to be as small as possible to minimize the impact
on cost.  The scheduler runs periodically to determine the
number of requests to be sent to each server in a given epoch to
achieve the target rate $T$ and low variance.

After chunks are downloaded, they are assembled into frames in the
playback buffer. As its name suggests, the {\em Bandwidth estimator}
(\S~\ref{sec:core}) uses the rate at which requested chunks arrive on
each connection to estimate the available bandwidth ($R_1, R_2, ...$)
to each server. The average bandwidth $\hat{R}_i$ to server $i$ and
its bandwidth variance $\Var(R_i)$ are reported to the scheduler as
inputs to its scheduling decisions.

After the scheduler determines the number of chunks that need to be
requested from each server, the information is sent to the {\em
  Request manager} (\S~\ref{sec:manager}). The request is in the form
$\alpha_1, \alpha_2, ...$, where $\alpha_i$ is the number of chunks to
download from server $i$. Note that chunks are of fixed size, and it
is the number of chunks requested from server $i$ (\ie $\alpha_i$)
that varies to adapt to current/predicted network conditions. We
discuss alternatives to this approach and their shortcomings
in~\S\ref{sec:alternative}.


In the rest of this section, we provide additional details on the main
components of the \Sys client, namely, the Scheduler, Bandwidth
Estimator, and Resource Manager.


\begin{figure}[hbt!]
\centering
\includegraphics[width=0.85\linewidth]{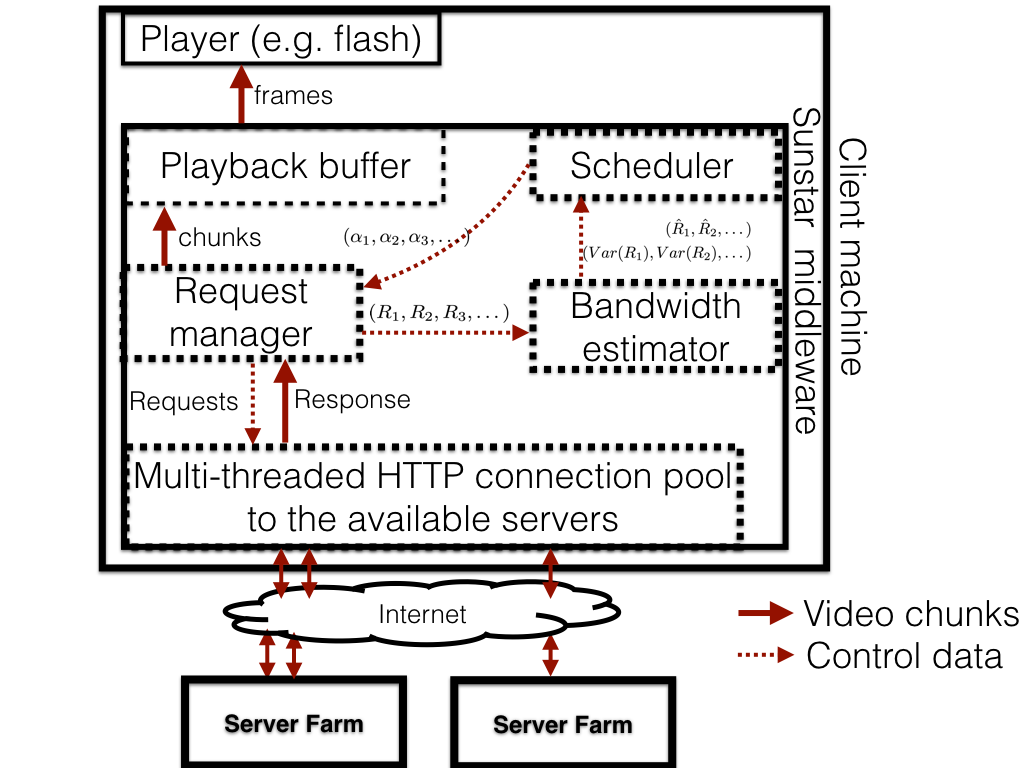}
\caption{\Sys architecture.}
\label{fig:architecture}
\end{figure}

\subsection{The \Sys Scheduler}
\label{sec:scheduler}

The scheduler is the core component of the \Sys
client\footnote{\S\ref{sec:eval} and \S\ref{sec:cost} show that
  it not only delivers significantly better performance than other
  schedulers, but also that it realizes those improvements with little
  to no impact on cost.}. Its main goal is to minimize playback stalls
(and therefore skips) as their frequency and duration are known to
have a significant influence on QoE~\cite{QoE1,QoE2} (see also
our Mechanical Turk experiments in Appendix~\ref{appendix:turk}).
Playback stalls are a direct consequence of an empty playback buffer,
which arises when the download rate falls below the playback rate for
an extended period of time. To minimize the odds of such occurrences,
the \Sys scheduler periodically runs an optimization that computes how
many chunks to request from each server to ensure a target average
download rate while minimizing variations around that average rate.

The optimization is greedy and myopic, \ie based on current
performance estimates for the paths to each server and does not
attempt to predict future performance. This is motivated by the fact
that it is difficult to predict bandwidth fluctuations ahead of time,
and a myopic algorithm can quickly correct course when conditions
change.

The \Sys scheduler operates in ``epochs'' or ``rounds'', whereby it
seeks to minimize the \textit{increase} in variance in each
round. This lets the optimization decide the number of requests for
each server based on short-term (per round) bandwidth estimates, and
adjust them in each round to respond to changes. An alternative would
rely on longer term bandwidth estimates, but this makes the algorithm
less responsive to occasional outliers (\eg network
failures/congestion events that happen every so often).
The round duration is set empirically (see~\S\ref{sec:exec} for
details) as a compromise between responsiveness, variability of the
estimates, and computation overhead (the optimization runs in each
epoch).  The optimization takes as input the current bandwidth
estimates to each server.  It then computes a target
rate (number of requests) for each server in the next round.  

The rest of this subsection describes the formulation of this
optimization in more details.

\subsubsection{Scheduler Optimization}
\label{sec:sched_optim}


Let $\mathcal{S}$ be the set of servers assigned to a MuMS
client. The goal is to guarantee each client an average
rate~$T$ (note that in the case of an ABR codec, this value would
change whenever the codec opted to adjust the rate based on its own
decision function), while minimizing \textit{changes} in the running
variance expressed as 
\begin{align}
(\sum_{i\in \mathcal{S}}{\alpha_{i} R_i} - T)^2,
\end{align}
where $R_i$ is the inverse of the time it takes for requests to arrive
from server $i$ (in other words $R_i$ is the inverse of the application
level round trip time), and $\alpha_i$ is the number of
chunks that the client should request from server $i$. The client's
attained rate is thus $\sum_i\alpha_i R_i$. 

Minimizing changes in the running variance then translates into
solving the following optimization: 
\begin{align}
\label{eq:opt1}
&\min_{\alpha} \mid \sum_{i \in \mathcal{S}} \alpha_i R_i - T \mid\\
&s.t. \quad  \sum_{i \in \mathcal{S}} \alpha_i \widehat{R}_i \ge T \nonumber\\
&\quad \quad \quad \alpha_i \le w_i \nonumber,
&\end{align}
where $\widehat{R}_i$ is the expectation of $R_i$, and $w_i$ the
current window size\footnote{$w_i$ upper bounds the number of requests
that can be sent to each server to enforce TCP-like congestion control
-- see~\S\ref{sec:core}.} to server $i$ (both computed by the
bandwidth estimator -- see~\S\ref{sec:core}).
The $\alpha_i$ are the decision variables and determine the number of
requests allocated to each server. Let $R_i^u$ and $R_i^l$ be upper
and lower bounds for $R_i$, respectively.  $R^u$ and $R^l$ can be
estimated using Chebychev's inequality by identifying the upper/lower
bounds on the $95$ and $5$ percentiles of the distribution of $R$.
\Eqref{eq:opt1} is then equivalent to:
\begin{align}
\label{scheduler}
&\min_{\alpha,t} \quad t  \\
&s.t. \quad \sum_{i \in \mathcal{S}} \alpha_i \widehat{R}_i \ge T \nonumber\\
&\quad \quad \sum_{i \in \mathcal{S}} \alpha_i \widehat{R}^u_i \le T+t \nonumber\\
 &\quad \quad \sum_{i \in \mathcal{S}} \alpha_i \widehat{R}^l_i \le T-t \nonumber\\
&\quad \quad \alpha_{i} \le w_i \nonumber.
\end{align}
\Eqref{scheduler} is derived
by first converting the absolute value form of the problem to its
linear form and then replacing the two resulting bounded constraints
with tighter bounds through $R_i^u$ and $R_i^l$.

The above formulation assumes integer values for $\alpha_i$, but
this can occasionally result in significant overshoots in the realized
rate.  We, therefore, use fractional values for $\alpha_i$.
Because requests are for an integer number of chunks, we
maintain a state variable that accounts for the ``excess'' rate $y_i$
to each server.  After solving the optimization, we compute
$\max(\alpha_iR_i-y_i,0)$ and use this value as the target rate to
server~$i$. When this corresponds to a fractional number of chunks, we
then round it up and update $y_i$ accordingly.

There are two other aspects to the above optimization that need
discussion, as they affect the implementation of the \Sys scheduler.

\begin{figure*}[hbtp!]
\centering
\begin{subfigure}{.5\textwidth}
  \centering
   \includegraphics[width=.75\linewidth]{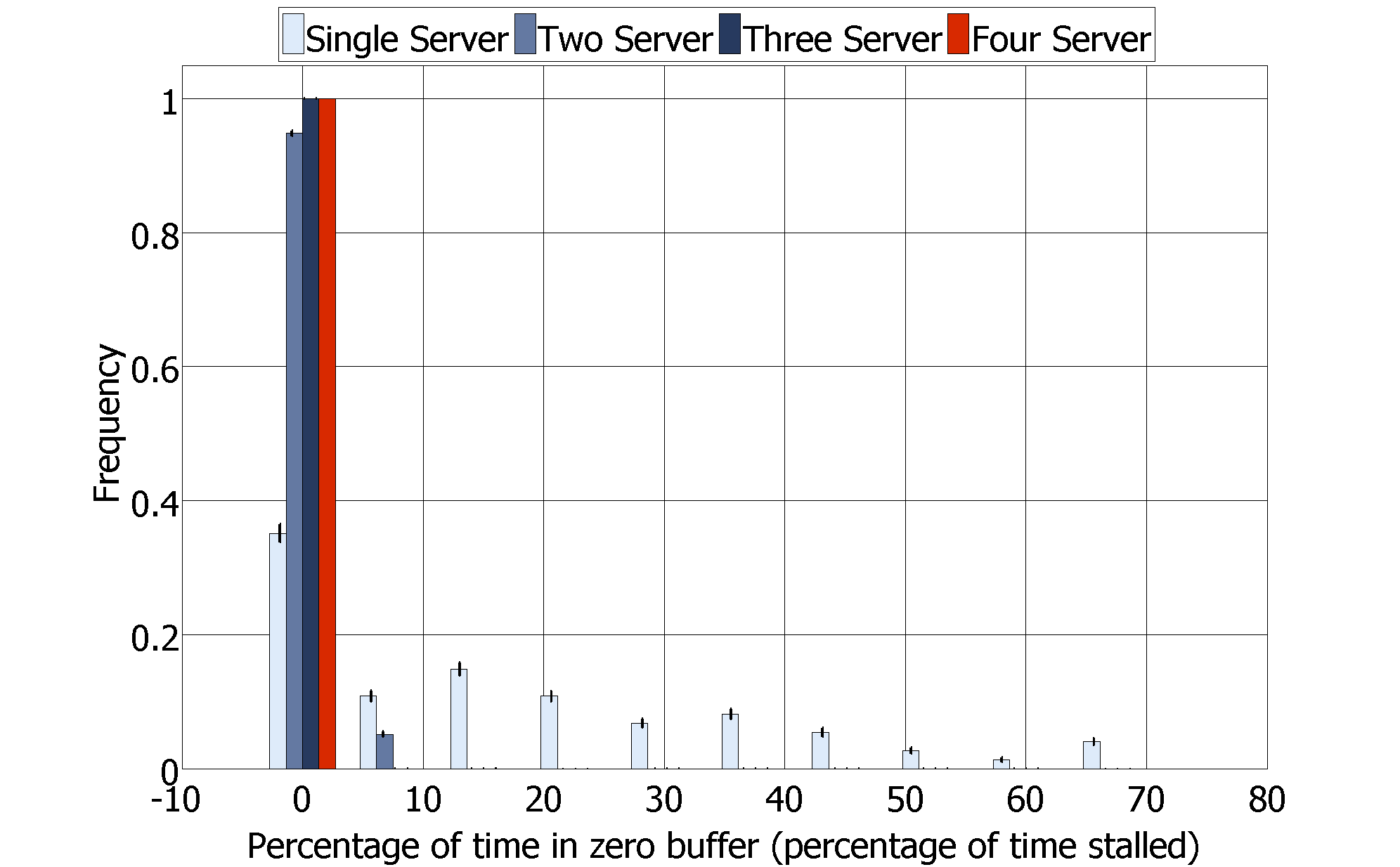}
\vspace{-1mm}
     \caption{Distribution of the fraction of time clients stalled in each run. \label{fig:stallsSinglePath}}
\vspace{-1mm}
\end{subfigure}%
\begin{subfigure}{.5\textwidth}
  \centering
\includegraphics[width=.75\linewidth]{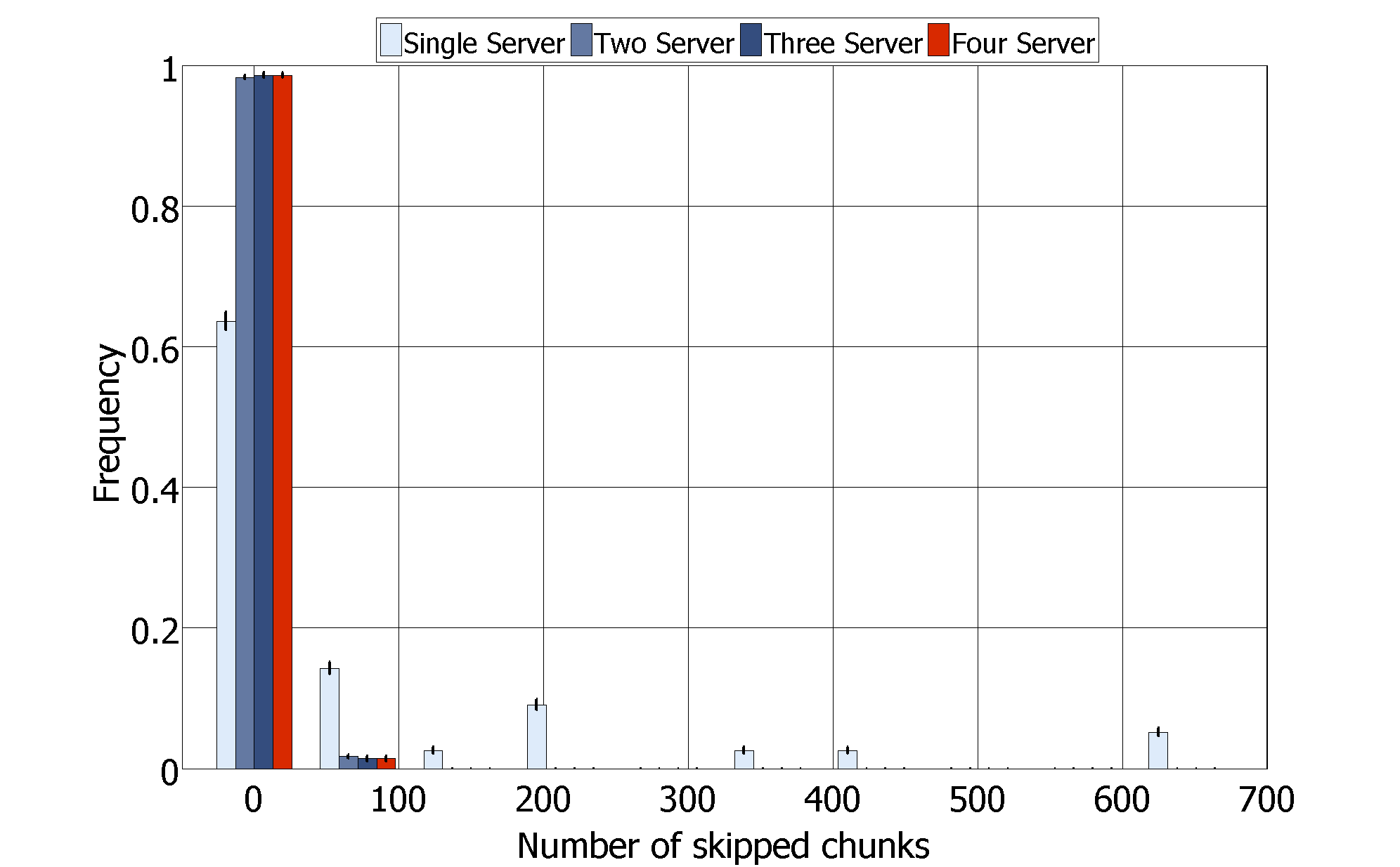}
\vspace{-1mm}
     \caption{Distribution of the number of chunks skipped in each run. \label{fig:skippedSinglePath}}
\vspace{-1mm}
\end{subfigure}
\vspace{-1mm}
\caption{Comparison to single server. {\bf High} and {\bf Medium}
  $C_i$, smooth bandwidth variations.} 
\vspace{-1mm}
\end{figure*}

The first is \emph{Dealing with Infeasibility.} As network bandwidth
fluctuates, the optimization may not always be feasible. However, we
still want the client to request the best possible transmission
rate\footnote{Alternatively, this may serve as an ABR trigger to
  switch to a lower rate.} (the closest rate to $T$). For that purpose,
we progressively decrease the target rate in the \Sys optimization
(by~$10\%$ in our experiments) until a feasible solution is
found. Once a feasible solution is found, \Sys attempts to bounce back
to its original target as quickly as possible by increasing its target
rate by $30\%$ in each subsequent round (while the solution remains
feasible) until it can again maintain its original target rate.


The second issue concerns \emph{Breaking ties.} The optimization needs
not have a unique solution, \eg with homogeneous paths with sufficient
bandwidth, using any of them is a feasible and optimal solution. We,
therefore, add two tie-breaking criteria to the optimization: (1) we
minimize the number of servers used, and (2) we favor those that have
been used more frequently in the past. Both criteria aim to reduce the
number of out-of-order chunks.

We show in Section~\ref{sec:exec} that the optimization of
\Eqref{scheduler} is practical in that it can be solved for up to $4$
servers in around $1$ms without overloading an entry-level
machine.
We also note that it is possible to improve the
optimization run time by re-using past solutions as a starting point
in each round~\cite{bertsimas}.  
Such improvements are, however, beyond the scope of this
paper.


\subsection{Bandwidth Estimation}
\label{sec:core}
It is important for the \Sys client to utilize the available
bandwidth effectively and not exacerbate congestion if/when it
happens. Thus, in our design of the \Sys client, we maintain an
estimate of the bandwidth available on each path using a TCP-like
mechanism, which computes a window size $w_i$ (in chunks) for each
path~$i$. 
This window provides an upper bound for the number of requests that
can be sent to a server. Specifically, $w_i$ is computed using a TCP
CUBIC congestion control mechanism.  We chose
CUBIC as our bandwidth estimation method as the congestion window
changes are solely dependent on the time of the last congestion event
and not the latency to the server (this is important as we use
multiple servers in the client and there might be large discrepancies
in the RTT to each of these servers).  Note that CUBIC relies on
losses to reduce its window size, and since chunks are requested over
{\tt http}, there is no application level loss in the \Sys client. We
use a drop of more than $20\%$ in the estimated average rate
$\widehat{R}_i$ to server~$i$ as equivalent to a CUBIC packet
loss. $\widehat{R}_i$ is computed using an exponential average of
request completion times (this acts as a low pass filter and
eliminates estimation noise). 
The $20\%$~threshold was chosen based on the best empirical
performance observed across a sweep of the possible values.

\subsection{The Request Manager}
\label{sec:manager}
As described earlier, the \Sys client is an application layer
protocol. Once the \Sys scheduler computes an allocation of requests
it informs the request manager of the new allocation. The request
manager maintains a list of chunks for which a request has not been
sent to any of the servers. It uses a multi-threaded HTTP connection
pool with open connections to each of the servers to request the
appropriate number of chunks from this list from each server.
Once responses arrive, the request manager ensures that it is placed in
the appropriate place in the playback buffer.

\subsection{Performance Optimizations}

The \Sys client incorporates optimizations previously devised for
MPTCP. In particular, it employs opportunistic retransmit. When
latencies across servers are significantly different, opportunistic
retransmit prevents the client from stalling while waiting for a chunk
requested from a high latency server.  A mechanism similar to TCP
time-outs is also implemented for chunk requests. If a response is not
received before a time-out, the request is re-sent to another server.
Last, we limit the number of retries for chunk requests to bound
playback stalls.  Once this retry limit is exhausted, the chunk is
skipped and the player relies on its codec to mask missing
frames. This decision is implemented as part of \Sys's request
manager.

Note that while, as mentioned earlier, \Sys is both complementary to
ABR and should be able to integrate with an ABR codec, the current
implementation is limited to fixed rate codecs.
Extending the design to make it fully compatible with ABR codecs is
part of future work.

\subsection{Discussion of Alternative Design Choices And Trade-offs}
\label{sec:alternative}

We acknowledge that there are a number of alternatives to our design
of \Sys. These include:

\noindent{\textbf{Fixed vs adaptive chunk sizes.}} Our design
uses fixed sized chunks and adjusts rates by varying the number of
chunks requested.  An alternative would be to keep the number of
chunks fixed and vary the chunk size (this is the
approach of~\cite{youtuber}).
We opted for the former approach as empirical evaluations with
both approaches showed that it exhibited better performance given our
goal of rate stability (See~\S\ref{sec:eval}).

\noindent{\textbf{Application vs transport layer solution.}}  An
alternative to an application layer solution is a transport layer one,
\ie by extending a protocol such as MPTCP. A transport layer solution
has advantages such as finer adaptation granularity and, therefore,
faster reactions. However, extending MPTCP to work in a multi-server
rather than single-server setting involves non-trivial changes to the
network protocol stack.  In addition, the tight coupling of MPTCP
components can result in unexpected interactions~\cite{arzani14}.
Avoiding or predicting them calls for a careful evaluation of proposed
changes. Hence, an application layer solution offers a number of
benefits, in terms of flexibility and ease of deployment.

\Section{Performance Evaluation}
\label{sec:eval}


\begin{figure}[htb!]
  \centering
\includegraphics[width=0.9\linewidth]{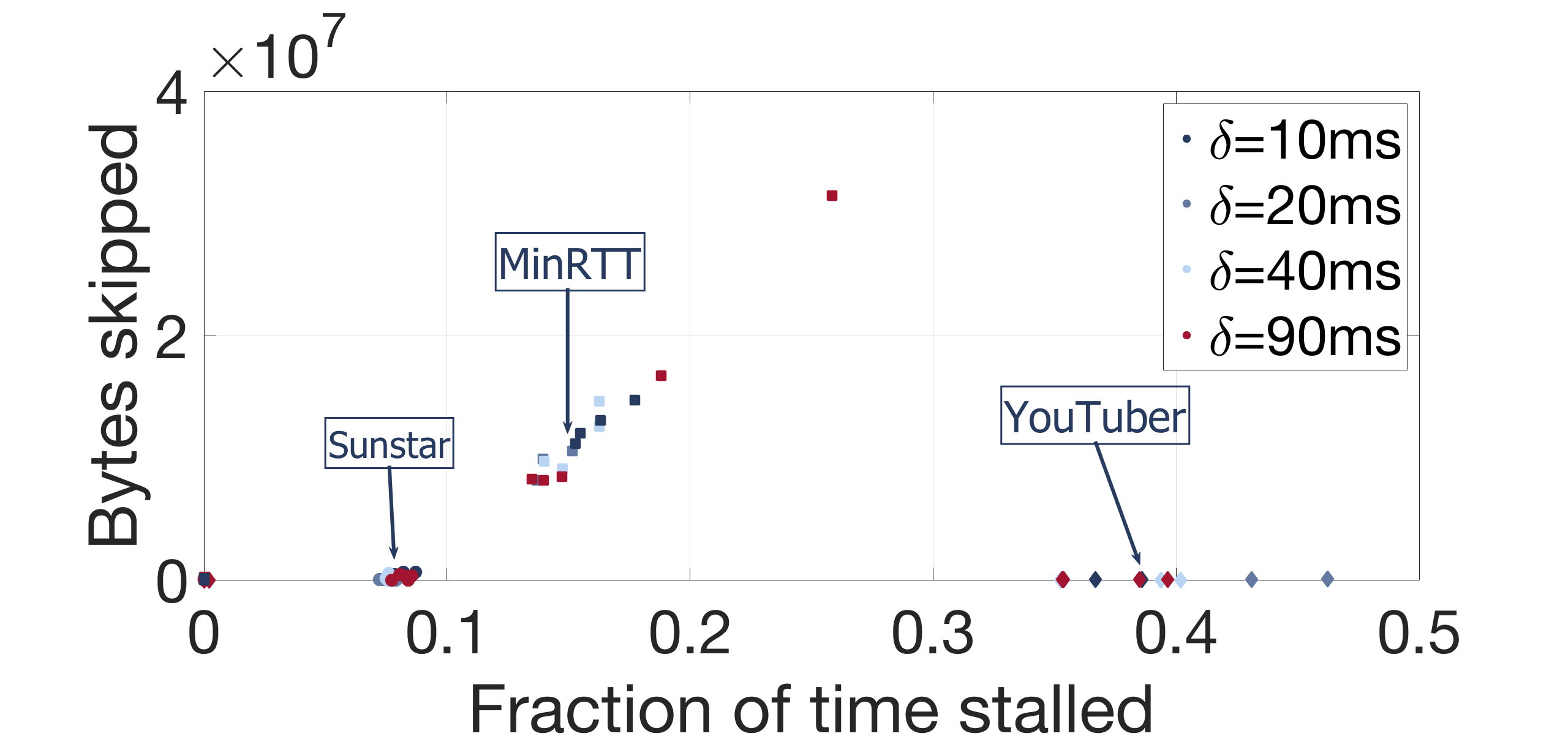}
     \caption{Impact of latency - Two servers, {\bf Medium}
       $C_i$, bursty bandwidth variations ($\delta=$ difference in
       latency).\label{fig:mismatched}}  
\end{figure}

We implemented a prototype of the \Sys client and in this section
report on its evaluation with respect to performance, namely,
(1) improvements in video quality over a single server client; (2)
improvements in video quality compared to two representative
multi-server schedulers: an RTT-based scheduler\footnote{We refer to
  it in the paper as the {\em Min-RTT} scheduler, and selected it as
  representative of schedulers that focus on performance, \ie by
  prioritizing downloads from the ``closest'' servers.} similar to
that used (at the transport layer) by MPTCP, and the YouTuber
scheduler\footnote{As previous results exist for YouTuber, its
  comparison to \Sys includes a few more scenarios than for the
  Min-RTT scheduler.}  of~\cite{youtuber} that works by estimating the
bandwidth available to each server, and requesting chunks at a rate
equal to that estimate; (3) its run-time performance, in particular,
that of its core optimization routine.  The other key aspect of the
\Sys design, \ie cost, is evaluated in~\S\ref{sec:cost}.

\begin{figure*}[htb!]
\centering
\begin{subfigure}{.47\textwidth}
  \centering
   \includegraphics[width=.8\linewidth]{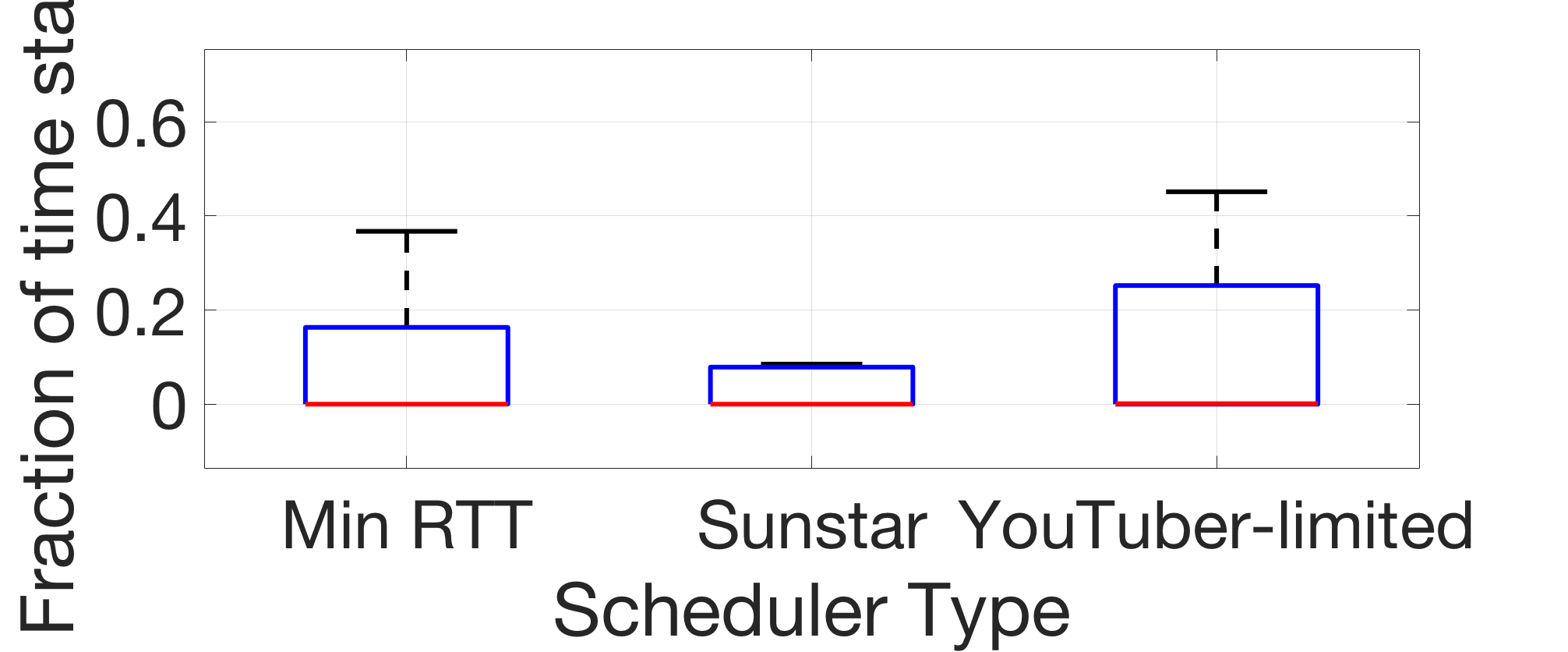}
\vspace{-1mm}
     \caption{Fraction of time stalled. \label{fig:LiveStreaming}}
\vspace{-1mm}
\end{subfigure}
\begin{subfigure}{.47\textwidth}
  \centering
\includegraphics[width=.8\linewidth]{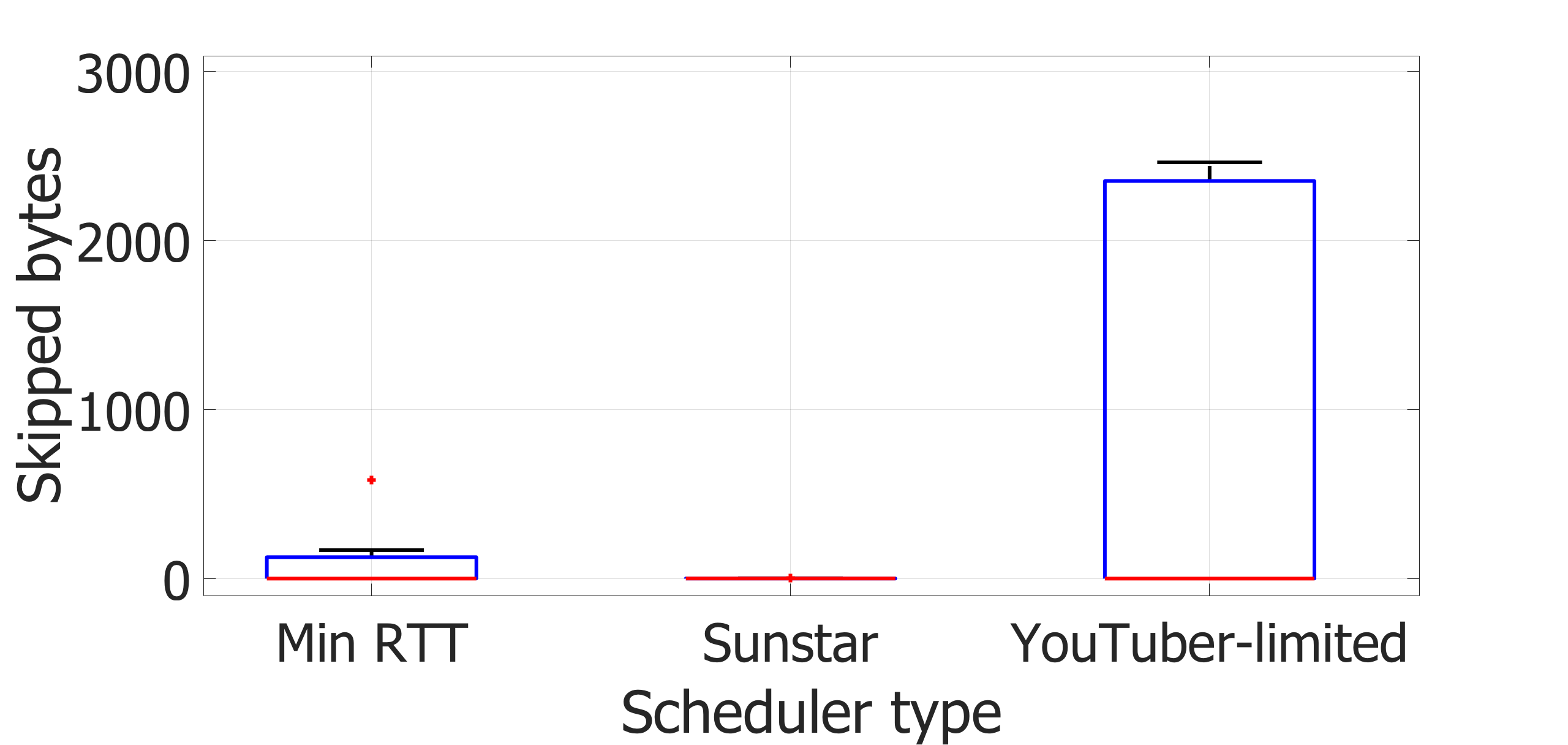}
\vspace{-1mm}
     \caption{Number of skipped bytes. \label{fig:youtuberskippedlive}}
\vspace{-1mm}
\end{subfigure}
\vspace{-1mm}
\caption{Live streaming schedulers comparison - {\bf Medium}
  $C_i$, bursty bandwidth variations.}
\vspace{-1mm}
\end{figure*}


\Subsection{Evaluation Setup}

Our evaluation setup is similar to that of \S\ref{sec:mums}, but
spans a broader set of scenarios to offer a more a comprehensive
evaluation of \Sys's performance.

Our \Sys prototype is a Linux-based MuMS client that retrieves video
from Nginx web servers. The videos are $250$~MB and broken into fixed
chunks of $102$~KB (results with slightly different chunk sizes and
with Apache servers were similar). The codecs are fixed rate codecs
with a target rate $T=4.08$~Mbps (experiments with $T=360$~Kbps
yielded similar results).  The scheduler's epoch is $10$~ms.

\begin{figure}[hbt!]
\centering
   \includegraphics[width=0.75\linewidth]{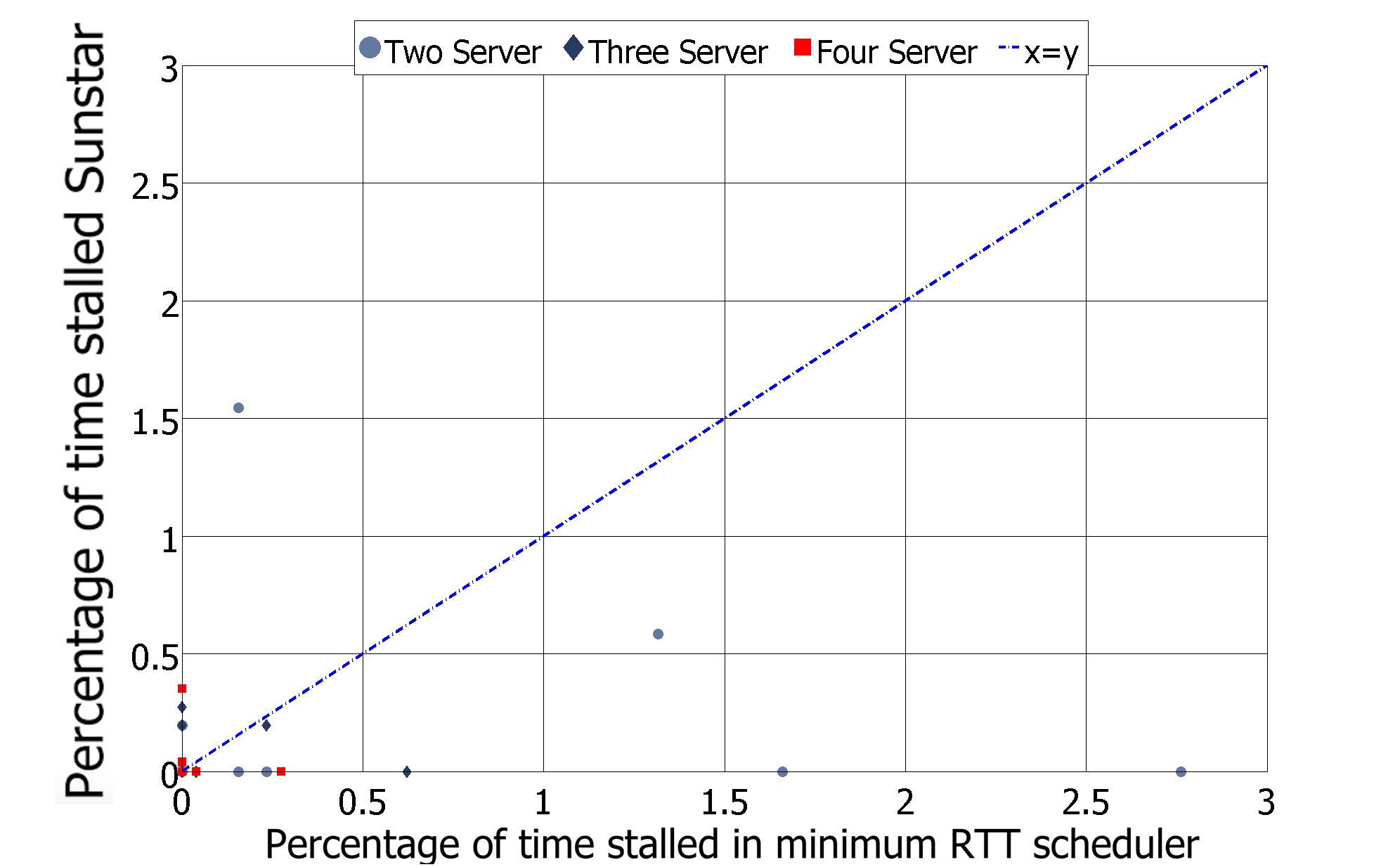}
     \caption{Min-RTT vs.~\Sys schedulers - {\bf High} $C_i$, smooth
       bandwidth variations.\label{stallsMultipathhigh}} 
\end{figure}
\begin{figure}[hbt!]
  \centering
\includegraphics[width=0.75\linewidth]{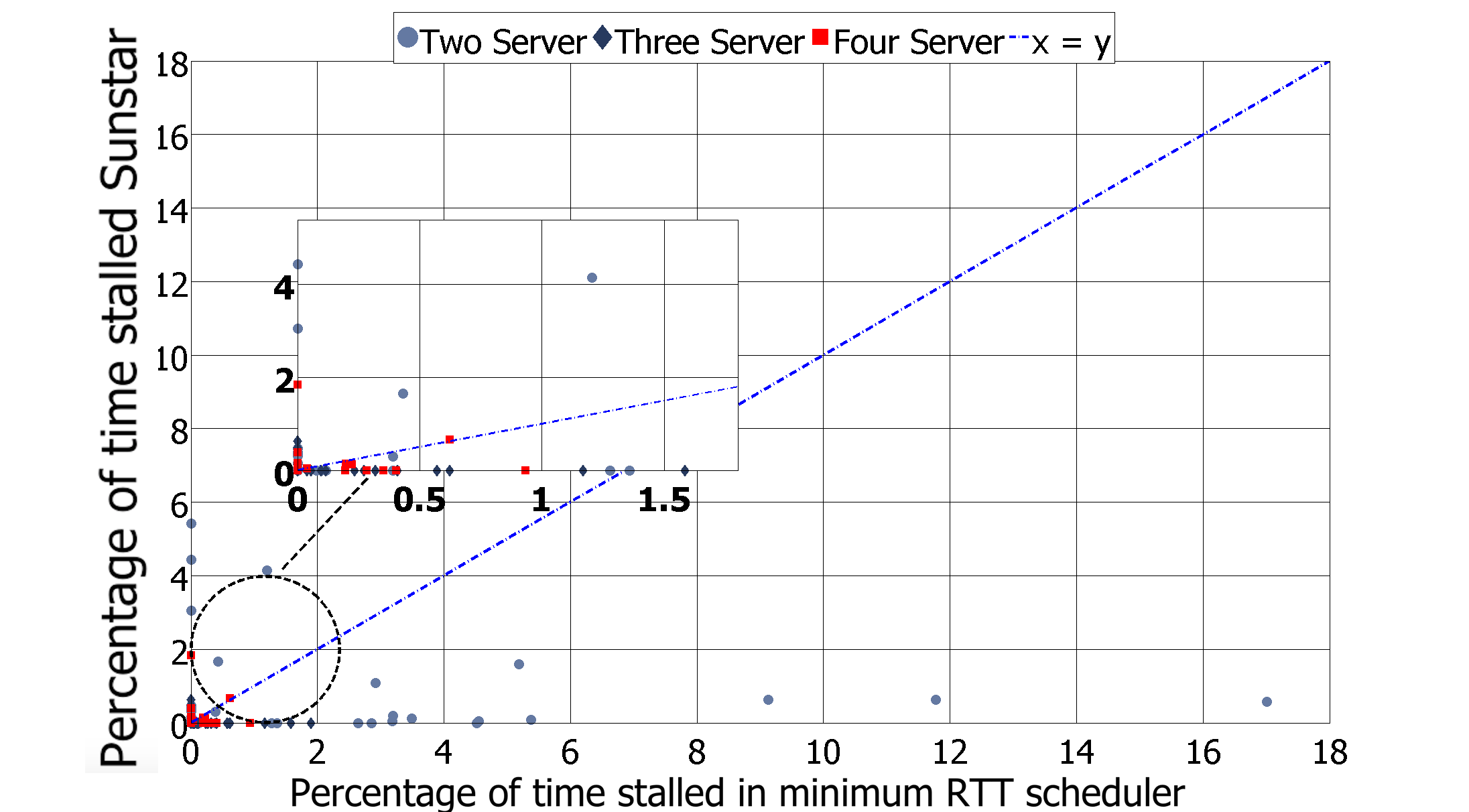}
     \caption{Min-RTT vs.~\Sys schedulers - {\bf Medium} $C_i$, smooth
       bandwidth variations.\label{stallsMultipathMed}}
\end{figure}

As before, experiments are carried out on the Emulab testbed. Clients
connect to each server via a dedicated ``path'' consisting of a link
connecting the client to a machine running dummynet, followed by a
link connecting that machine to the server.  The dummynet machine on
path~$i$ is used to add a fixed latency of~$20$ms and vary the
capacity available to the client on path~$i$ from~$0$ to~$C_i$ with an
average value of $C_i/2$. Those variations seek to capture the impact
of interfering traffic on the bandwidth available between clients and
servers.  We consider two types of bandwidth variations {\em smooth}
and {\em bursty}.  Under smooth variations, the available bandwidth
increases and decreases progressively in fixed and small sized steps
(we use a number of different step sizes for each experiment).  This
seeks to mimic the progressive bandwidth fluctuations that arise from
client arrivals and departures.  In contrast, bursty bandwidth
variations are based on large sporadic changes in available bandwidth
that represent abrupt changes in congestion, \eg because of the start
of a high-bandwidth download on a shared link or the start of a live
streaming event. 

We experiment with three configurations: (i) \textbf{High $C_i$},
where $T \ll C_i/2$. This captures the bandwidth
conditions observed in non-peak viewing times as well as that observed
by clients connecting to over-provisioned CDNs. (ii) \textbf{Medium
  $C_i$} where $T + 1$~Mbps $\le C_i/2 \le T + 2$~Mbps. This scenario
is more typical of today's ecosystem where CDN bandwidth provisioning
and server capacity are ``just enough" for the expected workload. In
such a setting, the bandwidth available to the client occasionally
falls short of its download rate, thus, negatively impacting its
playback experience. (iii) \textbf{Low $C_i$} where $C_i/2 < T$. This
is a scenario where the CDN and/or servers are oversubscribed, \eg
when a live streaming event is more popular than predicted. 
In the {\bf Low} $C_i$ scenario, multiple paths (servers), are necessary just to meet the target rate.



\noindent {\bf Performance metrics.} As in Section~\ref{sec:mums}, QoE
is measured based on~\cite{QoE1,QoE2,QoE3}: (1) the number of skipped
chunks; and (2) the fraction of times clients are stalled.

\Subsection{Comparison to Single Server Clients}

\begin{figure}[hbt!]
  \centering
   \includegraphics[width=0.8\linewidth]{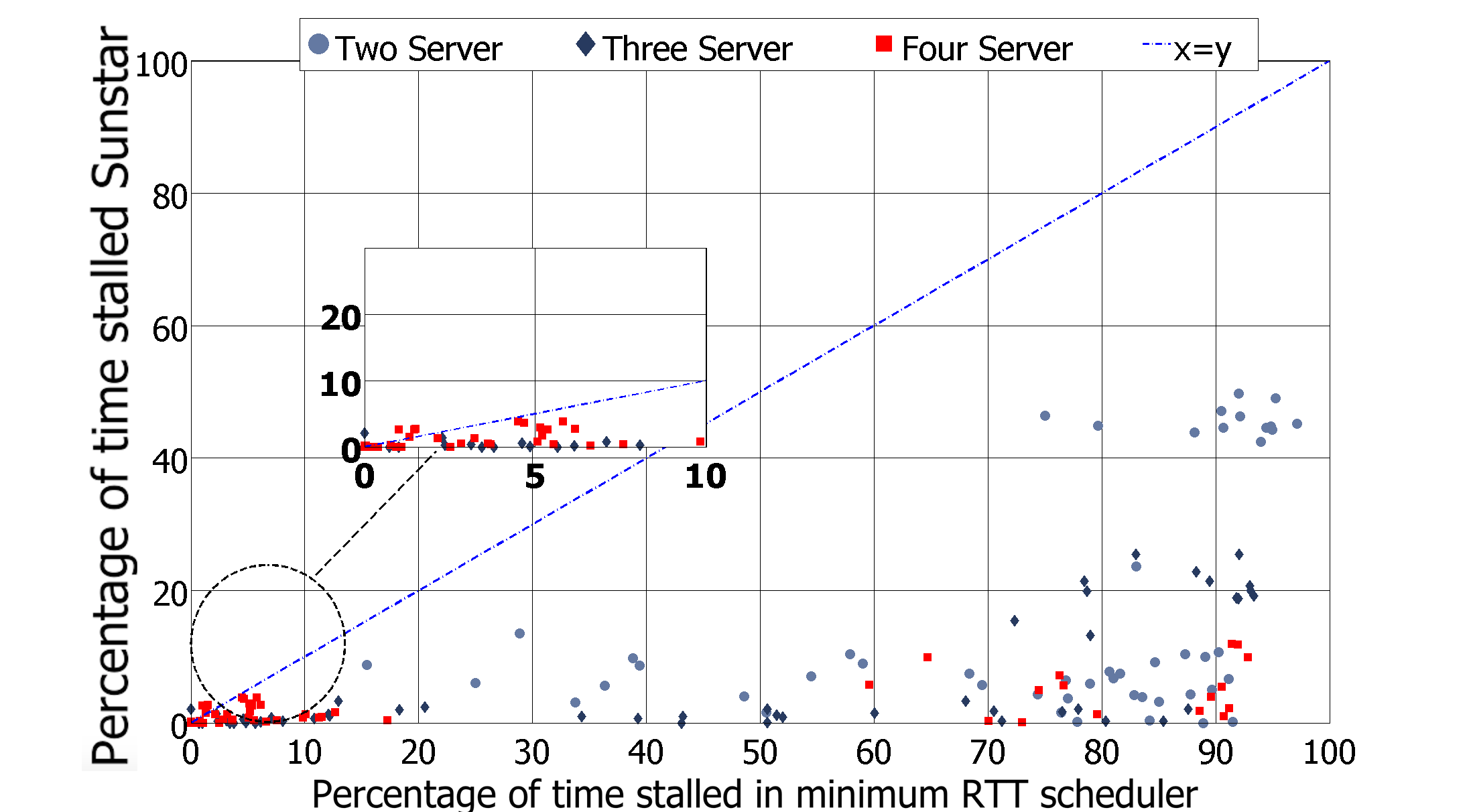}
     \caption{Min-RTT vs.~\Sys schedulers - {\bf Low} $C_i$, smooth
       bandwidth variations.\label{stallsMultipathLow}}
\end{figure}%
\begin{figure}[hbt!]
\centering
  \centering
   \includegraphics[width=0.70\linewidth]{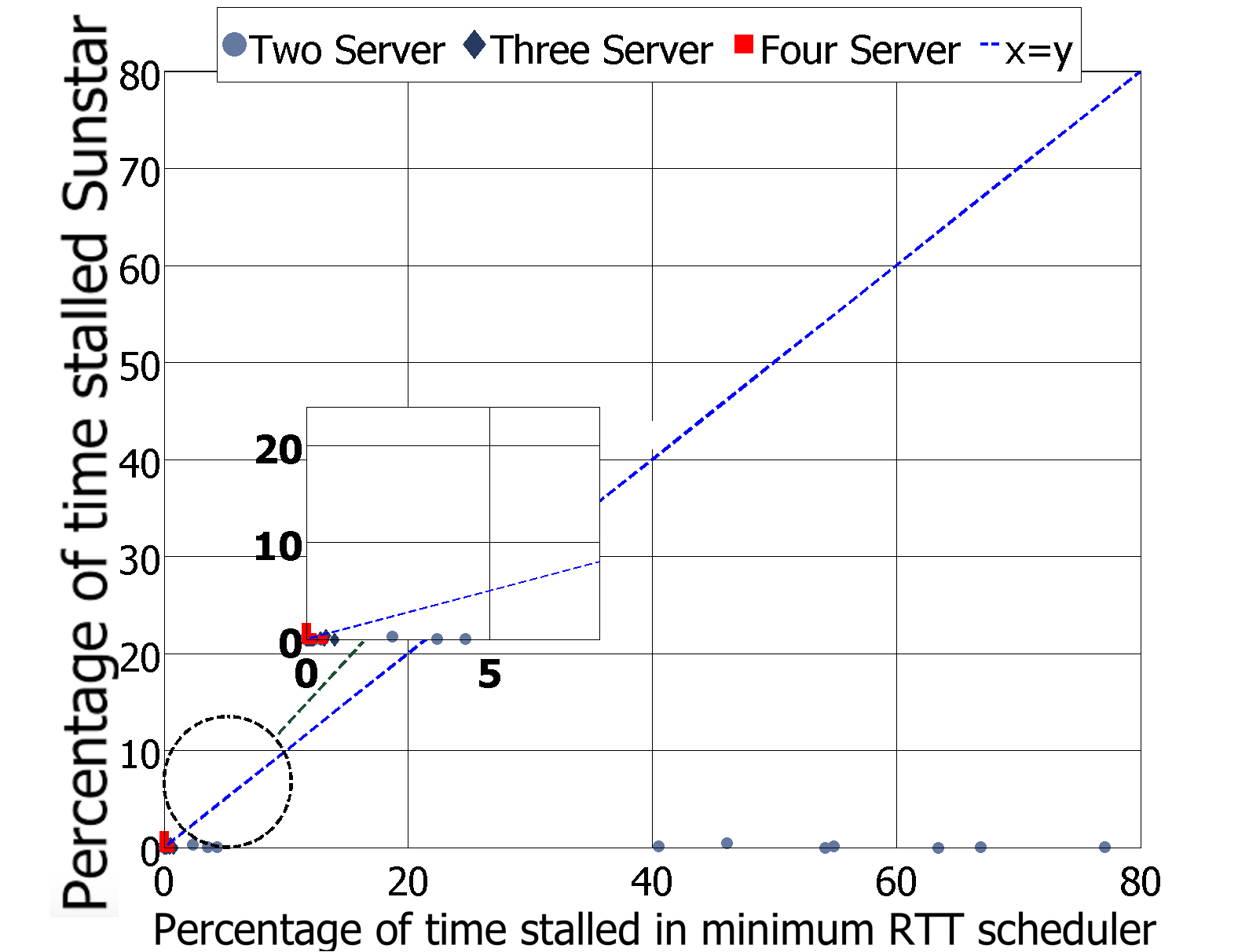}
     \caption{Min-RTT vs.~\Sys schedulers - {\bf Medium} $C_i$, bursty
       bandwidth variations.\label{fig:bursty}}
\end{figure}
We first compare the performance of the \Sys client to that of a
single server client (the typical streaming configuration today). This
repeats the earlier MuMS validation of \S\ref{sec:mums}, but now using
the \Sys scheduler instead of the Min-RTT scheduler.  Only the
\textbf{High} and \textbf{Medium $C_i$} configurations are considered,
since they are the only two for which an individual path has enough
(average) bandwidth. Experiments with smooth and bursty variations
yielded qualitatively similar outcomes, and we therefore report only
the former.
Statistics for stall durations and number of skipped chunks for
\textbf{High} and \textbf{Medium $C_i$} are combined and shown in
\figs{fig:stallsSinglePath}{fig:skippedSinglePath},
respectively\footnote{Error bars show the $95^{\mbox{th}}$ percentile
confidence intervals assuming Bernoulli distributed samples.}.  The
figures confirm the results of \S\ref{sec:mums}, but now for the
\Sys client.  A casual comparison of \fig{fig:skippedSinglePathMPTCP}
and \fig{fig:skippedSinglePath}, also hints at the \Sys scheduler
out-performing the Min-RTT scheduler.
We explore this aspect next.




\begin{figure*}[hbt!]
\centering
\begin{subfigure}{.5\textwidth}
  \centering
   \includegraphics[width=.8\linewidth]{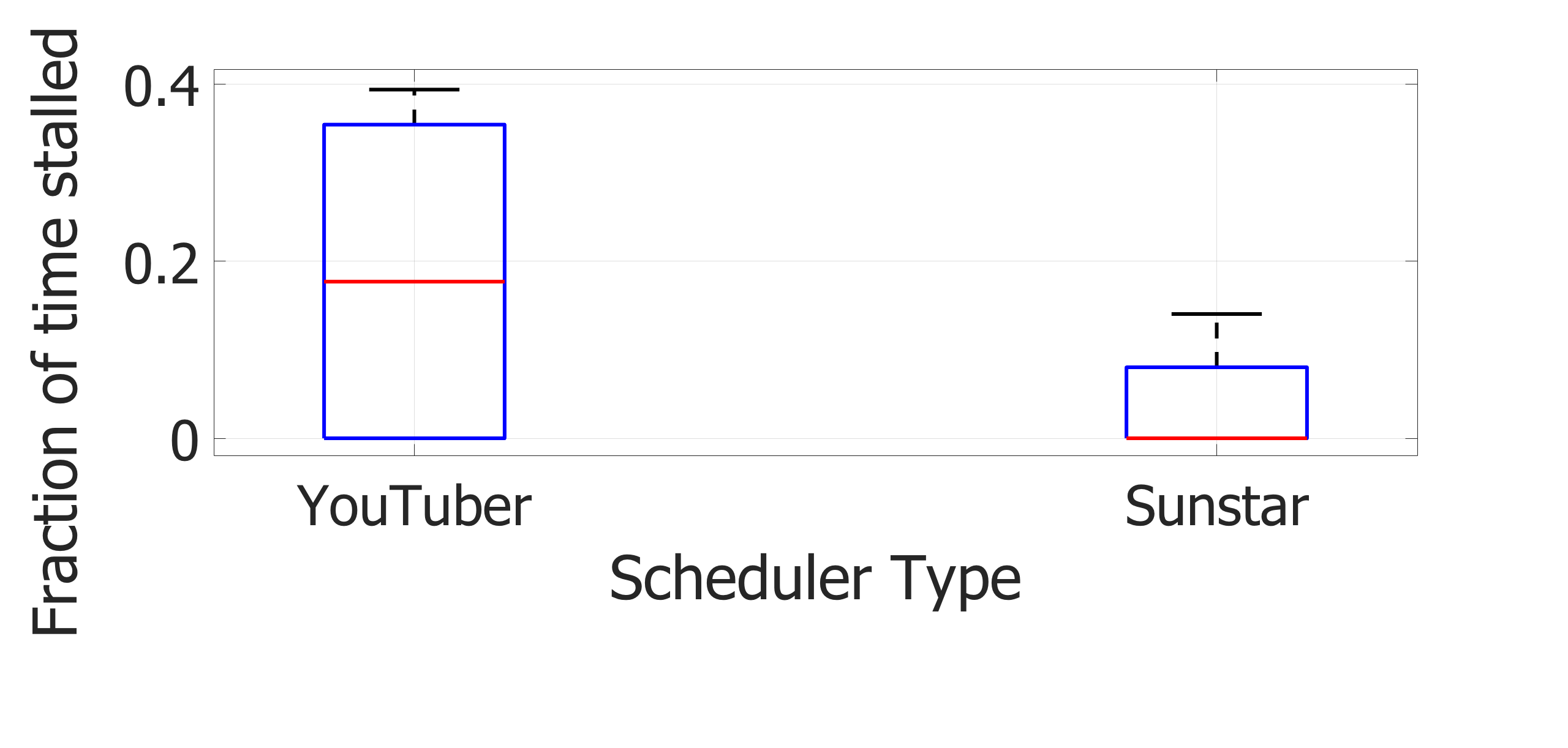}
     \caption{Fraction of time stalled.\label{fig:youtuberfractionstalled}}
\end{subfigure}%
\begin{subfigure}{.5\textwidth}
  \centering
\includegraphics[width=.8\linewidth]{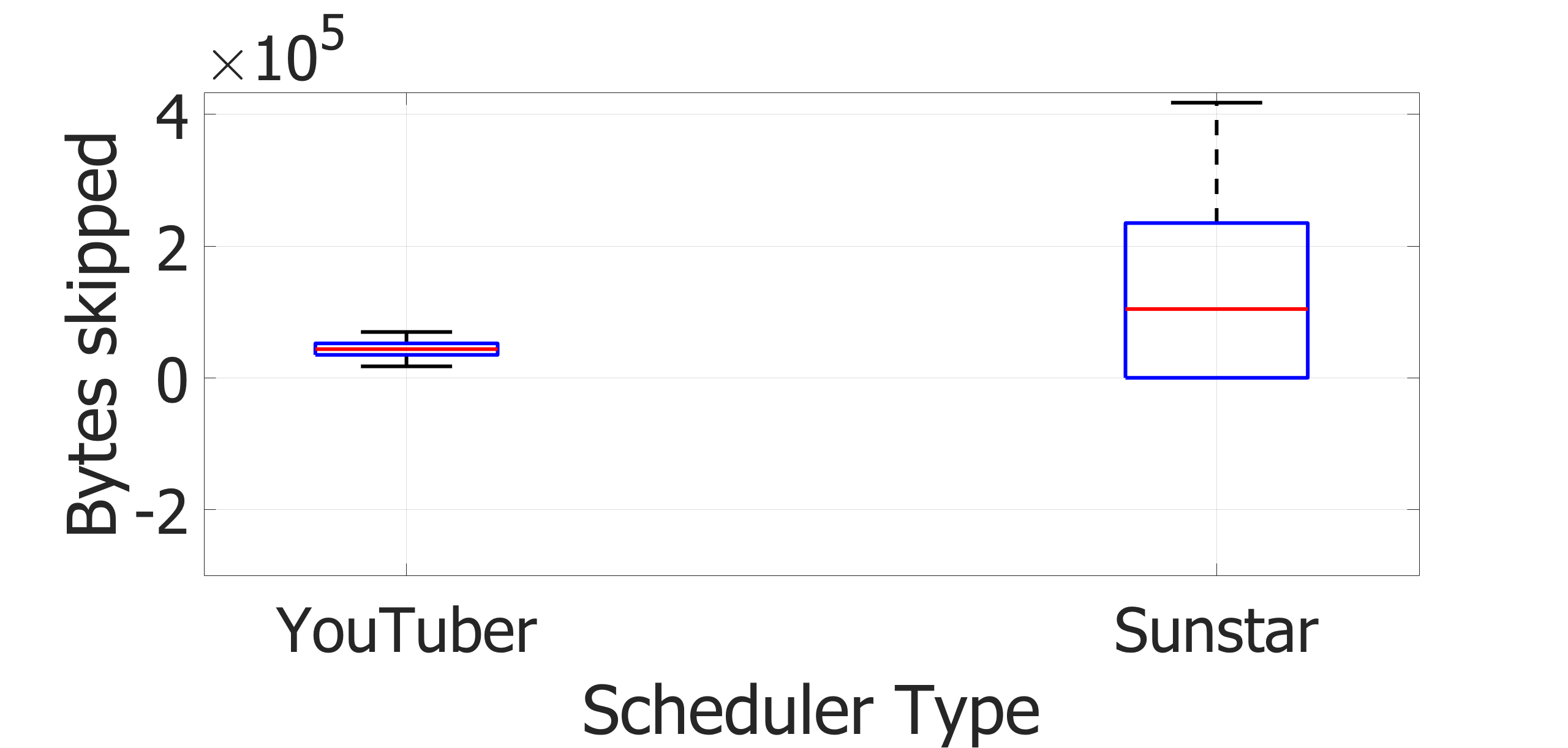}
     \caption{Number of skipped bytes. \label{fig:youtuberskipped}}
\end{subfigure}
\caption{\label{fig:youtuber_vs_sunstar}YouTuber vs.~\Sys schedulers - {\bf Medium}
  $C_i$, bursty bandwidth variations.}
\end{figure*}

\Subsection{Comparison to the Min-RTT scheduler}

We focus on a configuration where a single client downloads videos
from a given set of servers (from 2 to 4). 
The client uses either the \Sys scheduler or the Min-RTT
scheduler. Due to the high variance in the performance of the Min-RTT
scheduler, we report stall statistics\footnote{Results for skipped
  chunks statistics were of a similar nature.} in the form of scatter
plots (the large confidence intervals for the Min-RTT scheduler make
the results hard to interpret).  The $x$-coordinate of each data point
corresponds to the Min-RTT scheduler and the $y$-coordinate to the
\Sys scheduler.  Points below the $x=y$ line, therefore, indicate
better performance for \Sys.  Results are shown in
\figsm{stallsMultipathhigh}{stallsMultipathLow} for the {\bf
  High, Medium} and {\bf Low} $C_i$ configurations under smooth link
variations. \fig{fig:bursty} presents results for one representative
configuration with bursty bandwidth variations, namely, the {\bf
  Medium} $C_i$ configuration.


The figures show that
\Sys consistently outperforms the Min-RTT scheduler
irrespective of the number of servers used.  The biggest
improvements arise in the \textbf{Low $C_i$} scenarios, where the
limited resources amplify the need for judicious scheduling.
The \textbf{Medium $C_i$} scenario still sees \Sys
outperforming the Min-RTT scheduler, while the differences are
less pronounced in the \textbf{High $C_i$} scenario because the
plentiful resources ensure that both schedulers perform well.

\Subsection{Comparison to The YouTuber Scheduler}
The next set of experiments compare stalls and skips statistics of the
\Sys and YouTuber schedulers, again for scenarios where clients
download from a given set of servers.  As the MinRTT scheduler, the
YouTuber scheduler focuses on maximizing client performance. It
achieves this objective by matching its download rate to the available
bandwidth.  Since the original design of~\cite{youtuber} only
considers two servers, we limit our comparison to this scenario.  In
addition, because the YouTuber scheduler uses variable sized chunks,
we compare the number of skipped bytes instead of the number of
skipped chunks. YouTuber also lacks an explicit time-out mechanism. To
avoid penalizing it we, therefore, add a time-out period of $10s$ to
its operation.  When a stalls exceeds the time-out, $16$KB of data
(the minimum chunk size) is skipped.  Finally, In reporting the
results, we focus on {\bf Medium} $C_i$ configurations that are more
representative, and bursty bandwidth variations as rapid changes are
expected to stress both schedulers' adaptability\footnote{In
  configurations with smooth bandwidth variations \Sys and YouTuber
  displayed comparable performance.}.

\figs{fig:youtuberfractionstalled}{fig:youtuberskipped} report stalls
and skips statistics, respectively. The results are in the form of
boxplots showing the results for each metric separately.
The box boundaries correspond to the $25^{\mbox{th}}$ and
$75^{\mbox{th}}$ percentiles, with the median as the line inside the
box. The whiskers correspond to outliers.  

The results indicate that under bursty bandwidth variations, the \Sys
scheduler outperforms the YouTuber scheduler in terms of stall
statistics, but is slightly worse when it comes to skips statistics.
This latter difference is mostly because the YouTuber's original
design did not provide for time-outs, and the modification we applied
relies on a relatively large (10s) time-out.  This choice was
motivated by the fact that YouTuber relies on variable size chunks and
can at time request very large chunks.  A low time-out value would
then have often resulted in unnecessary skips for those very large
chunks.  A 10s time-out avoided this problem and limited the number of
skips that YouTuber incurred, but at the cost of a higher frequency of
stalls as seen in \fig{fig:youtuberfractionstalled}.  A complete
solution to the problem likely calls for an adaptive time-out
mechanism based on chunk size.  However, designing and implementing
such a solution is beyond the effort we could invest towards extending
the YouTuber design.  The results of \fig{fig:youtuber_vs_sunstar}
offer a representative sample of how the YouTuber and \Sys schedulers
compare in terms of performance, with both typically offering better
performance than, say, the MinRTT scheduler.
However, as we shall see in \S\ref{sec:cost}, \Sys's benefit also
extend to affording those performance improvements without impacting
the (peering) costs of the video provider.







\Subsection{Impact of Latency Heterogeneity} 
The previous experiments assumed identical propagation delays to all
servers. In this section, we briefly test how heterogeneity in
propagation delays affects the schedulers' efficacy.  For simplicity,
we limit ourselves to a scenario with two servers, {\bf Medium} $C_i$,
and bursty bandwidth variations.  The difference $\delta$ in
propagation delays on the paths to the two servers varies from $10$ms
to $90$ms. The results are reported in \fig{fig:mismatched}.  As the
figure shows, the Min-RTT scheduler exhibits the worst performance,
which degrades as $\delta$ increases.  YouTuber is again successful at
avoiding skips, but this is because of the relatively large time-out we
configured and at the cost of frequent stalls.
\Sys performs well with low stalls and skips statistics and a relative
insensitivity to $\delta$.
\Subsection{Live Streaming} 
Although video downloads represent the bulk of video traffic, live
streaming remains an important service.  It is, therefore, also of
interest to explore the benefits of a MuMS solution for live
streaming.  Note that unlike video downloads, live content is
generated at a constant rate, so that pre-buffering options are
limited. We, therefore, emulate a live streaming experience by first
writing a number of bytes equal to the client's pre-buffering
threshold into a file.  Subsequently, the server continues to write at
a constant rate of $T$ to the file until the entire file is
created. Clients start sending requests after the pre-buffer portion
of the file has been created.  Our experiments compare the Min-RTT,
\Sys, and a modified\footnote{The YouTuber client was not originally
  designed for live streaming, and its assumption of an unlimited
  playback buffer would often result in infeasible requests. Initial
  live streaming experiments with the original YouTuber design
  confirmed its poor performance.  We, therefore, modified it by
  introducing a limited playback buffer that avoided many of those
  problems. In all fairness, there might be better modifications to
  allow it to accommodate live streaming.} YouTuber (YouTuber-limited)
schedulers, for live streaming clients. We focus on a {\bf Medium}
$C_i$ configuration with bursty bandwidth variations.


The results are shown in
\figs{fig:LiveStreaming}{fig:youtuberskippedlive}, which illustrate
that \Sys outperforms the other schedulers in both stalls and
skips. YouTuber-limited has the worst performance, in part because, in
spite of our modifications, its aggressive download strategy results
in unnecessary stalls. Its design aims at maximizing its download rate
during ``on'' periods (when the playback buffer content drops below a
pre-specified threshold).
In live streaming, such an aggressive download strategy can result in
requesting content before it is available.  This results in a stall
while waiting for a retransmission (of the request).  Limiting the
YouTuber buffer size, as we did, allows the client to pace itself (by allowing the playback buffer to fill up and the player to go into ``off" mode), but
is still not entirely successful at eliminating unnecessary stalls.

We repeated the above experiments with larger pre-buffers.
A larger pre-buffer creates a larger ``margin'' to absorb subsequent
rate variations, at the cost of a delayed start in the live stream.
This should benefit all schedulers, but particularly YouTuber, as it
can help mitigate occurrences of stalls.  As expected, performance
improved for all schedulers, but \Sys continued to outperform the
other two.

\Subsection{Scheduler Execution Time}
\label{sec:exec}

Last, we evaluated \Sys's run time performance. Recall
that \Sys's optimization runs every epoch ($10$ms in our
experiments), so that its efficiency matters.

The optimization uses the Mosek Solver~\cite{mosek}, and the Emulab
client machines on which it runs are Dell PowerEdge 2850s with a
single 3GHz processor, 2GB of RAM, and two 10,000 RPM 146GB SCSI
disks~\cite{pc3000}.

Based on our experiments, the run time varies (increases) with both
the number of servers and with the amount of bandwidth available to a
client on each path (both contribute to larger search spaces for the
optimization).  We varied the number of servers from two to four and
the (average) bandwidth available to individual clients on each path
from $3$~Mbps to $10$~Mbps. The fastest average run time was for a two
servers, low bandwidth configuration, for which it was $0.27\pm
0.0022$~ms. The longest average run time was for four
servers in a high bandwidth scenario, where it was $2.20\pm
0.0026$~ms.  In both configurations, the margins are the
$95$~percent confidence intervals.  The run time increased
significantly when using more than four servers, suggesting that
without additional computational optimizations, \eg offloading
computations to the cloud, four servers likely represents a realistic
limit.  We also experimented with increasing the epoch duration, but
while some increases are possible beyond our initial value of $10$ms,
\Sys's performance degrades rapidly as it increases further.  This is
because larger values limit \Sys ability to adapt to bandwidth
variations.

\Section{Effect on Peering Costs}
\label{sec:cost}

\Sys primary motivation is to improve video quality, but do so
without negatively impacting peering costs.  The insight derived from
Section~\ref{sec:mums} led to a design that minimizes rate
variations while keeping the download rate as close as possible to the
minimum feasible rate, \ie the target download rate.
The previous section showed this was effective in improving video
quality.  The focus here is on establishing that those benefits are
realized without increasing the provider's cost.  There are two
separate aspects to the impact on cost.  The first is the effect of
the scheduler on rate variations on individual peering links that
affect the $95^{\mbox{th}}$ percentile.  The second is the influence
of server selection on the load of individual peering links. We
explore both separately in this section.

We evaluate peering costs using a setup similar to that of
Section~\ref{sec:eval}.  The main differences are increases in both
the number of clients simultaneously active, and the number of servers
available to them (we now have $10$ servers to chose from). The latter
allows us to consider the impact of server selection on cost.  While
this section is only concerned with cost, we also evaluated \Sys's
performance and verified that its benefits remain qualitatively
similar to those of Section~\ref{sec:eval}.

\begin{figure}[hbt!]
  \centering
\includegraphics[width=0.8\linewidth]{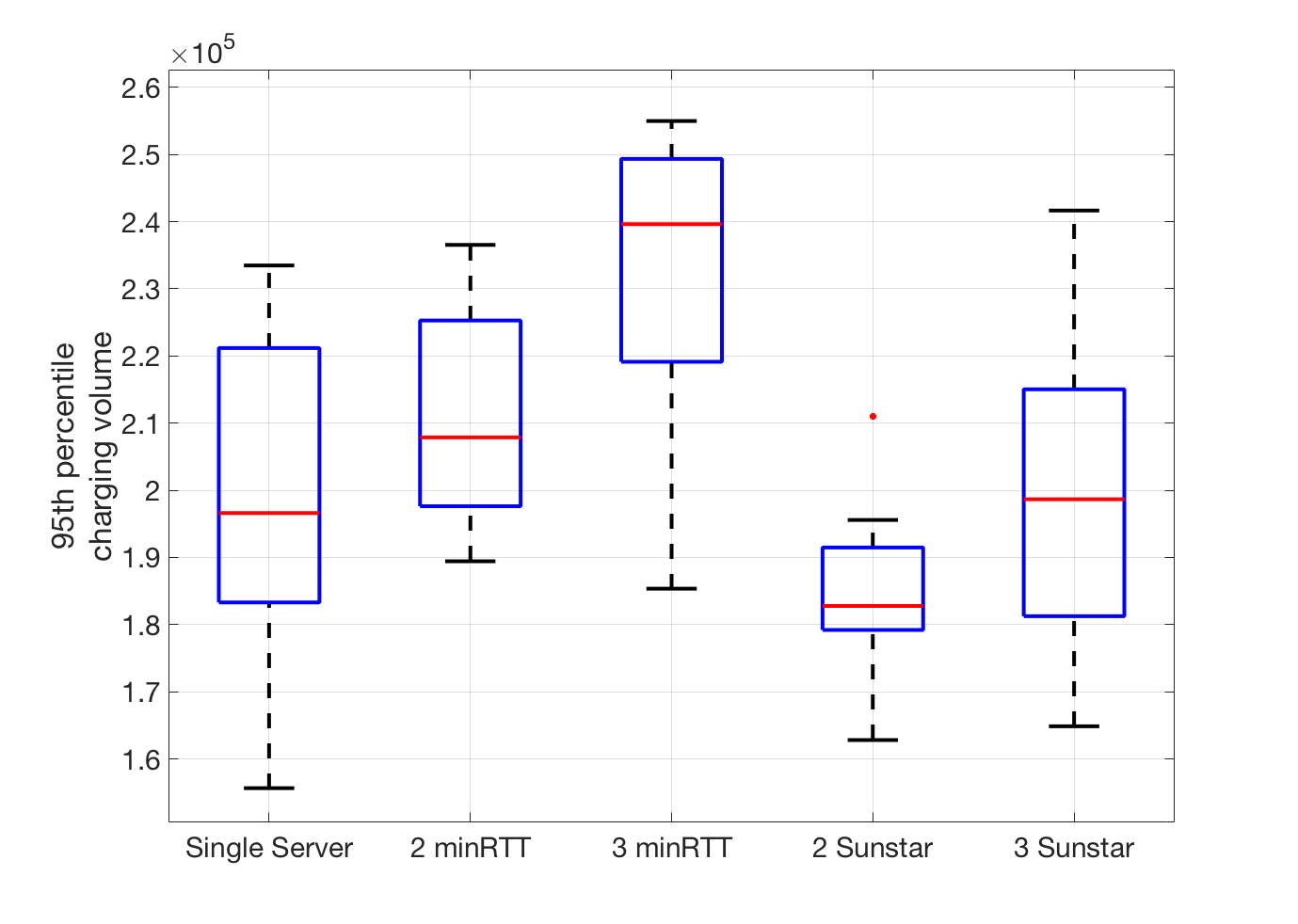}
     \caption{Peering costs under different schedulers using round robin server selection.\label{fig:cost_schedule}}  
\end{figure}

\subsection{Scheduler Impact}

Because the YouTuber scheduler behaves aggressively when it comes to
download rate (it seeks to use as much of the available bandwidth as
possible), we expect (and validated) that it performs poorly when it
comes to cost.  We, therefore, focus our efforts on comparing \Sys to
a single server solution (baseline), and to a client using the Min-RTT
scheduler. 

In a given round of experiments, clients connect to $k$ out of $10$
Emulab servers.  Since Emulab has a limited number of physical
machines available, we configure each machine to have $10$ active
clients at any point in time, for a total of $100$~clients in the
system.
To emulate an environment with clients coming and leaving, clients
watch videos of fixed duration chosen from a set of
$\{5,10,20,30,60\}$ minutes long videos, and then leave to be replaced
by a new client that randomly chooses a new video.  Video selection is
biased towards shorter videos (based on the observations
of~\cite{gill2007youtube}).

Each physical machine has a dedicated link to a dummynet node
through which clients originating on that machine experience bandwidth
variations that are independent of those for clients on other machines.
We use a {\bf Medium} $C_i$ configuration, as described in the
previous section, but scale the bandwidth by a factor $10$ (to account
for the number of clients on the link).  Low and high bandwidth
scenarios yielded qualitatively similar results in terms of cost.
Each server is in turn logically connected to a single peering link
shared by all clients accessing it.  The bandwidth on the peering link
itself is high enough to avoid congestion, independent of the number
of clients assigned to the server.
New clients first connect to a ``master'' server, which redirects them
to a list of $k$ servers, from which to download their video.  In this
section, the master uses a simple round-robin server
assignment strategy to select which $k$~servers ($k \in \{1,2,3\}$) to
assign to a new client.

For comparison purposes, a given experiment uses the same link
bandwidth variation patterns and server assignments for all
schedulers.  An experiment spans $4$~hours, with the provider cost
obtained by summing the individual $95^{\mbox{th}}$ percentiles
peering costs (traffic volumes) of the $10$ servers in those
$4$~hours.  Statistics are then computed over a set of $10$
independent experiments.  
\fig{fig:cost_schedule} reports results for
the following configurations: Single Server, our baseline, Min-RTT with 2 and 3 servers, 
\Sys with 2 and 3 servers.  The definition of the boxplots used in
the figure is similar to that of Section~\ref{sec:eval}.

\begin{figure}[hbt!]
  \centering
\includegraphics[width=0.8\linewidth]{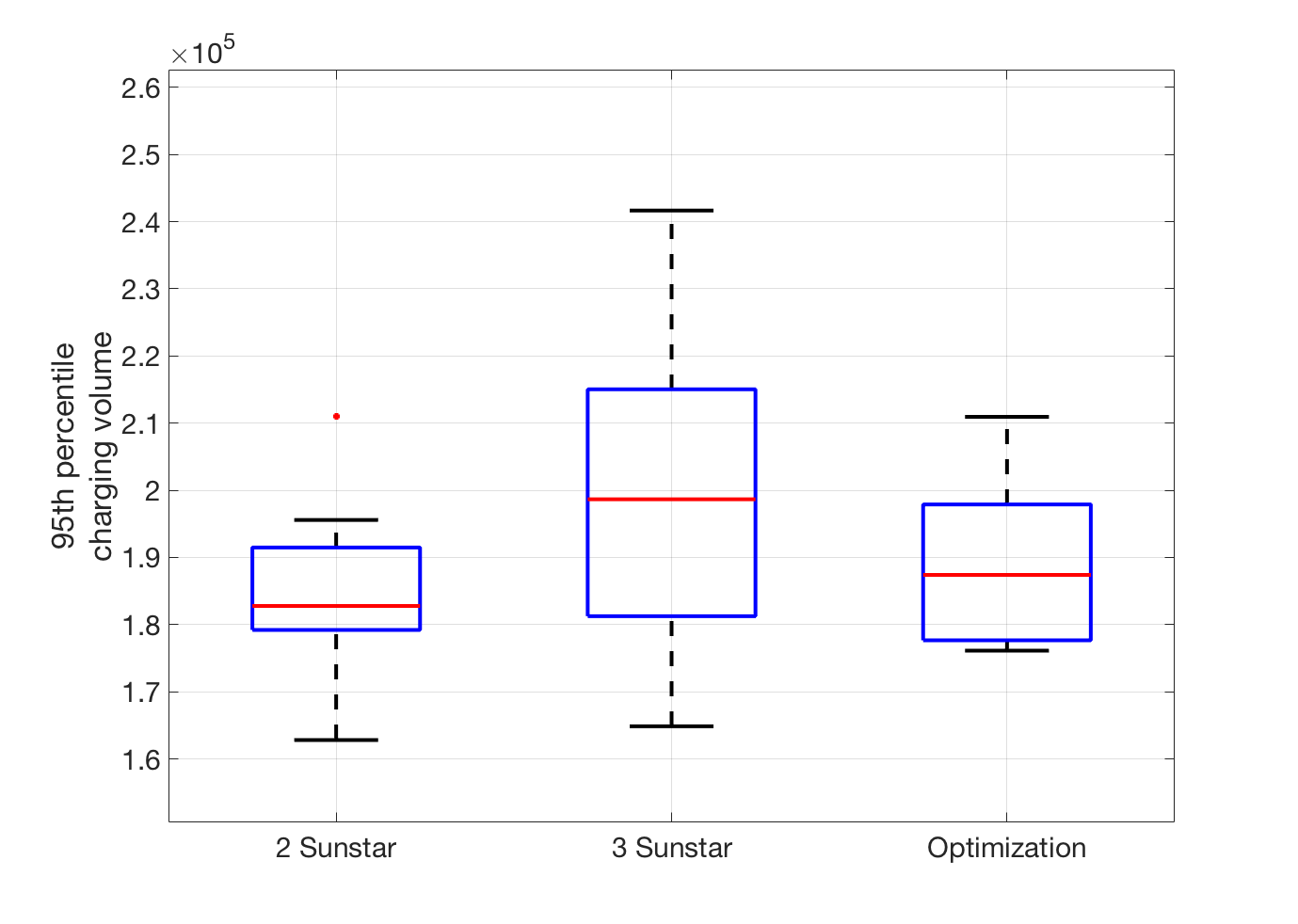}
     \caption{Peering costs comparison with ``smart'' server selection scheme that jointly optimizes for cost and performance.\label{fig:cost_opt}}  
\end{figure}

We make two observations from the results of \fig{fig:cost_schedule}.
The first is a confirmation of the insight of
Section~\ref{sec:why_cost}, namely, \emph{cost increases with the
  number of servers.} This is seen in the figure for both the min-RTT
and the \Sys schedulers, which display cost increases as the number of
servers goes from $2$ to $3$, even if the magnitude of the increase is
slightly less for \Sys.  The latter leads to our second observation,
namely, \emph{\Sys's rate variation minimization strategy is
  successful in mitigating cost increases.}  Specifically, \Sys
$2$~servers configuration outperforms not only a $2$~servers solution
using the Min-RTT scheduler, but also the baseline single server
configuration. And \Sys $3$~servers configuration has a cost
comparable to the single server baseline, and significantly lower than
its Min-RTT counterpart.  This offers empirical validation of the
simple analysis of Section~\ref{sec:why_cost} and the guidelines it
inspired.  In other words, \Sys succeeds in improving video
performance without affecting provider's cost.



\subsection{Server Selection Impact}

This section focuses on exploring whether a server selection strategy
that jointly optimizes for cost and performance can help further
reduce peering costs.  Since to the best of our knowledge no prior MuMS
server selection design exists that explicitly targets minimizing
provider cost, we
explore whether a server selection algorithm 
that jointly optimizes for cost {\em{and}} performance can help \Sys
further reduce its cost.  We formulate the server assignment problem
as a constrained optimization\footnote{Note that unlike the
  optimization of the \Sys scheduler, this optimization is required
  only once when a new client starts.} (see
Appendix~\ref{appendix:game} for details) that seeks to greedily
minimize increases in cost when assigning servers to new clients,
while meeting client rate constraints.

The results from experiments combining this optimization with the \Sys
scheduler are shown in
\fig{fig:cost_opt}, which compares the cost of \Sys for $2$ and $3$
servers under the previous round-robin assignment policy, to its cost
using the results of the optimization of Appendix~\ref{appendix:game}.
The figure illustrates that optimizing the server assignment did not
yield a meaningful reduction in cost, with the outcome falling in
between the results for the $2$~servers and $3$~servers scenarios.
This is not unexpected since the \emph{number} of servers that the
optimization assigns to a given client is not fixed.  In particular,
the optimization may assign any number of servers to a new client (up
to the maximum number available) when warranted by performance.
Because of space limits, we do not report the performance of the
combination of \Sys and our server assignment optimization, but it did
not offer statistically significant improvements in performance.  In
other words, optimizing server assignment did not help the \Sys client
achieve either better performance or lower cost, when compared to a
simple round-robin assignment policy\footnote{A separate experiment
  using a server selection strategy that picks the $k$ closest servers
  (lowest RTT), yielded a similar outcome.}.  The latter is partly due
to the \Sys scheduler design, as its goal of keeping rates low (as per
the results of Theorem~\ref{theo:cost}) realizes much of the available
gains in cost reductions.

This being said, we acknowledge limitations in the optimization of
Appendix~\ref{appendix:game}.  In particular, its reliance on a greedy
approach to estimate the impact of a new server assignment on the
$95^{\mbox{th}}$ percentile cost can clearly be improved, albeit at
the cost of significant added complexity.  A solution that better
predicts the impact of different assignments on costs may, therefore,
further improve on our formulation.

\Section{Conclusion}
\label{sec:concl}

The paper presented the design and implementation of \Sys, a MuMS
client aimed at improving video quality \emph{without} impacting
providers' peering costs.  \Sys relies on insight developed from
simple models for evaluating the effect of both multiple paths and
download rates on peering costs.  Those models helped illustrate why
both multiple paths and aggressive download rates could negatively
affect providers' costs.  This led to a design based on an
optimization framework that guarantees clients a target rate, while
minimizing rate variations.  Experiments on Emulab demonstrated \Sys's
ability to successfully meet its goals.

\appendix

\subsection{Mechanical Turk Experiment}
\label{appendix:turk}




\begin{figure*}[thb!]
\centering
\begin{subfigure}{.33\textwidth}
  \centering
   \includegraphics[width=.9\linewidth]{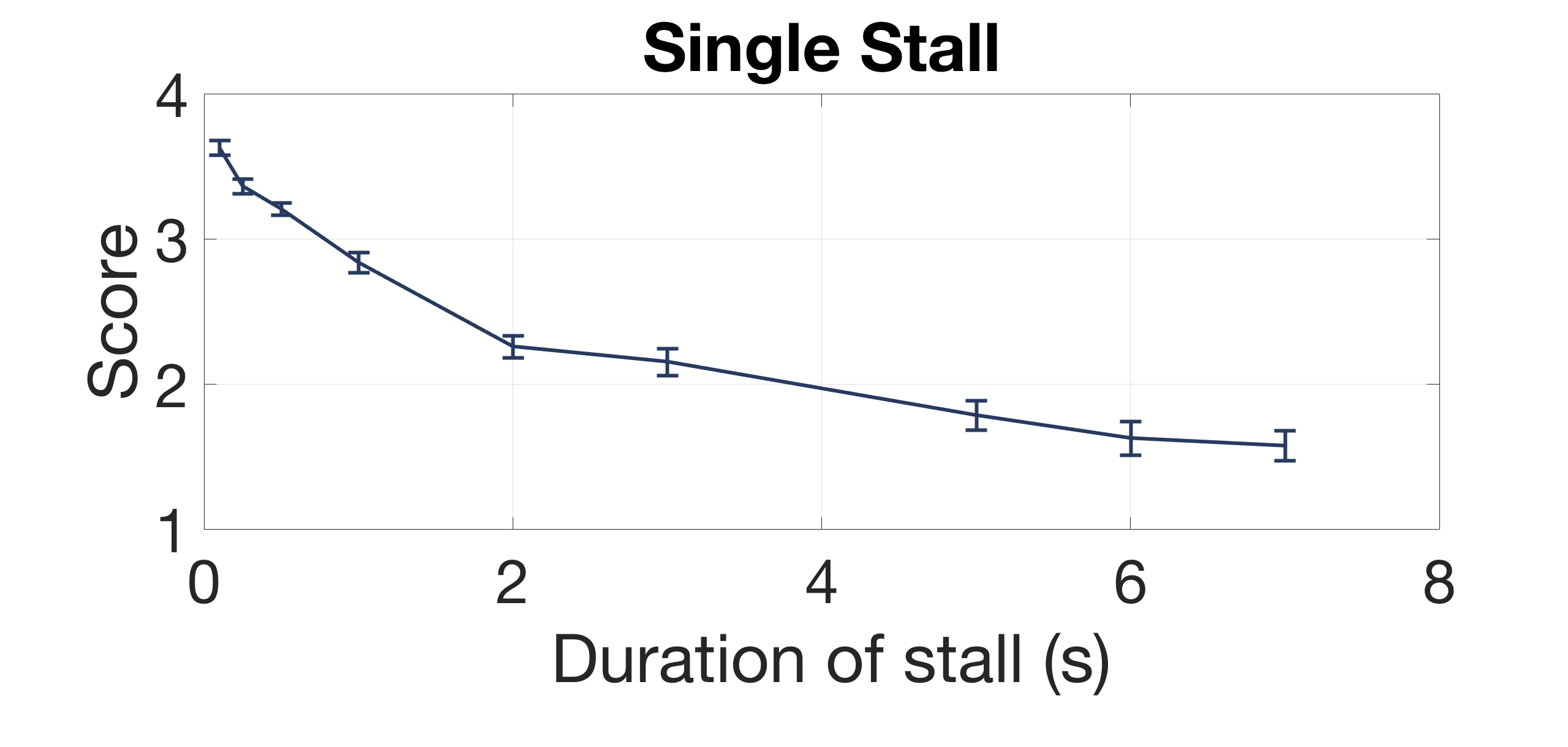}
     \caption{Single stall.}
\end{subfigure}%
\begin{subfigure}{.33\textwidth}
  \centering
\includegraphics[width=.9\linewidth]{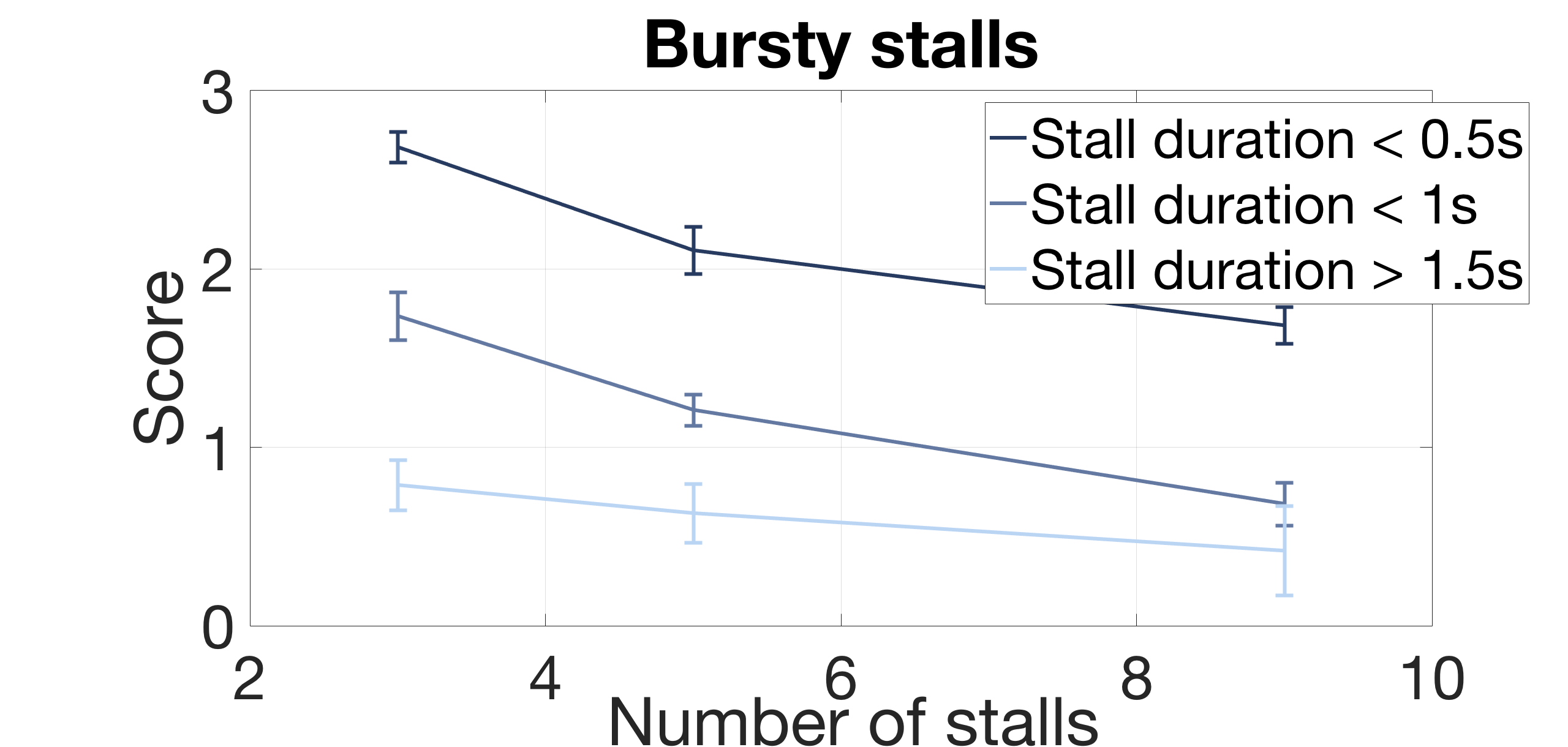}
     \caption{Multiple stalls, spaced out.}
\end{subfigure}
\begin{subfigure}{.33\textwidth}
  \centering
\includegraphics[width=.9\linewidth]{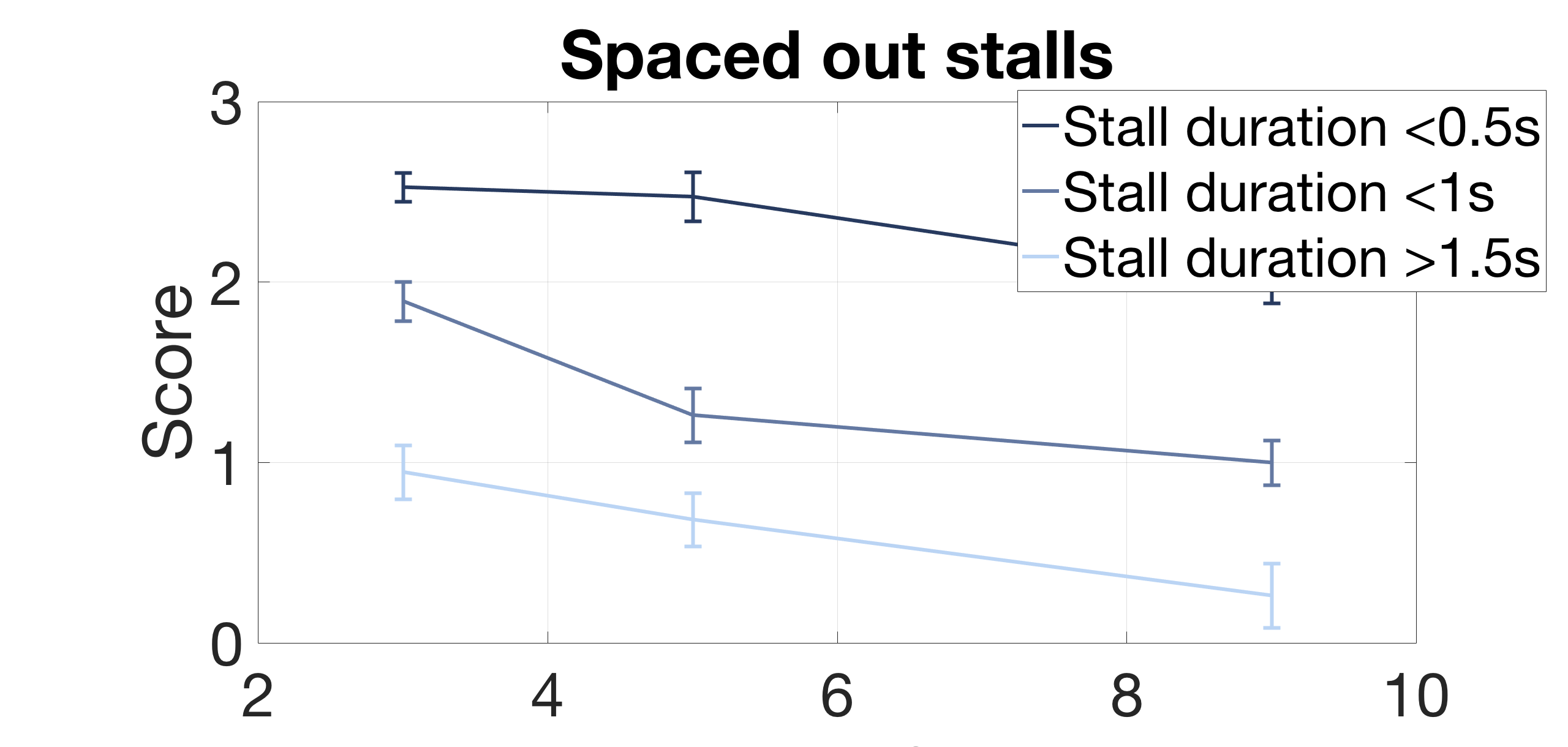}
     \caption{Multiple stalls, bursty.}
\end{subfigure}
\caption{Qualitative (Mechanical Turk) analysis of video
  QoE. \label{fig:mechturk}} 
\end{figure*}

The motivation was to further validate the QoE metrics we selected to evaluate video quality.  For that purpose, we
used a high quality (HD) documentary about Buckingham
Palace\footnote{\url{https://www.youtube.com/watch?v=jffKwoWjXtg}.}.
The video was divided into equal sized segments of $1$~min each, and
different types of impairments were introduced in those segments.
Due to of logistics constraints, only results for stalls are
available.  Specifically, we considered: 1) a single stall of variable
duration at a random location in the video; 2) multiple stalls of
small ($\approx 0.5$~sec), medium ($\approx 1$~sec), and long
($\approx 1.5$~sec) durations, evenly distributed in the segment; 3)
multiple stalls with the same distribution in duration, but now
closely spaced ($0.1$~sec) in a burst.  In 2) and 3) we varied the
number of stalls.  The quality of the video segments was evaluated on
a $0-5$ scale ($0$ being the lowest quality) by $20$ users recruited
through Amazon's Mechanical Turk market.  For calibration purposes,
users were first presented with an unimpaired video segment, and told
to assign it a rating of~$5$.
\vspace{-2mm}

Results of the study are presented in \fig{fig:mechturk}, which
confirms a strong correlation between stalls, both number and
duration, and video quality.  The limited size of the study is clearly
insufficient for broad conclusions, but it further confirms previous
QoE studies~\cite{QoE1,QoE2,QoE3} and the impact of stalls on video
quality.  Hence, fewer/shorter stalls do translate into higher video
quality.

\subsection{Server Selection Algorithm}
\label{appendix:game}

Server selection can be viewed as a Stackelberg game between the
clients and the provider, with the provider as the leader and clients
as the followers.  Once assigned servers, clients seek to maximize
their performance by scheduling requests to servers accordingly.
Given this behavior, the provider's goal is to assign servers so as to
minimize the
$95^{\mbox{th}}$ percentile cost.  This non-convex cost function
together with the online nature of the game make
computing the optimal assignment strategy hard.

We therefore propose a semi-online greedy optimization that is run
every $5$~mins and uses the current estimate of the $95^{\mbox{th}}$
percentile cost to assign client's to servers in a way that meets
their rate guarantees while minimizing cost. Specifically, the
optimization maintains an estimate of the number of client's expected
to arrive from each region (clients in a region have similar bandwidth
profiles and share the same connections to servers). Given these
estimates, it seeks to identify which assignment of servers for each
group of client results in the smallest increase in the current
$95^{\mbox{th}}$ percentile cost.  Furthermore, while the
optimization's goal is to minimize peering costs, it acknowledges that
this should not be at the expense of poor performance for the
clients. Thus, it also includes two additional constraints:
\vspace{-1mm}
\begin{align}
E\left[\sum_j \alpha_{ij} R_{ij} - T\right] \ge 0\\
E\left[\left(\sum_j\alpha_{ij} R_{ij}-T\right)^2\right] \le \gamma 
\label{sec:constraint2}
\end{align}
where $\alpha_{ij}$ is the number of requests client $i$ sends to
server $j$ and $R_{ij}$ is the rate in chunks per second from that
client's region to server $j$. 

Note that, reusing the notation of Section~\ref{sec:sched_optim},
\Eqref{sec:constraint2} can be written as $a_iTQa_i+b^Ta_i+c \le
\gamma$ where $Q$ is a matrix with $Q_{ii}=\Var(R_i)+\widehat{R}_i^2$ and
$Q_{ij}=\widehat{R}_i\widehat{R}_j$, $b$ is a vector where
$b_i=-2T\widehat{R}_i$, and $c=T^2$.

Take $\mathcal{F}_j$ as the current $95^{\mbox{th}}$ percentile cost
on peering link $j$, $\mathcal{L}_j$ the current load on peering link
$j$, and $m_i$ the expected number of clients arriving from region $i$
in the current decision period. We aim to solve the following
optimization:
\begin{align*}
&\min_{\alpha_{ij}} \quad \max_{B_j} (B_j+\mathcal{L}_j -\mathcal{F}_j,0)\\
&s.t. \quad B_j = \sum_i m_i\sum_j \alpha_{ij}\widehat{R}_{ij}\\
&\quad\sum_{j} \alpha_{ij}\widehat{R}_{ij} \ge T_i \quad \forall i\\
& \quad\alpha_{ij} \le w_{max}\\
&\quad\alpha_{i}^TQ_i\alpha_i+b^T\alpha_i+T^2 \le \gamma^2\\
& \quad \mathcal{B}_j+\mathcal{L}_j \le C_j
\end{align*}
where $w_{max}$ is the maximum window size allowed on the clients, and
the server selection algorithm assigns all servers with $\alpha_{ij} >
0$ to region $i$. It is straightforward to show that $Q$ in the above
equations is positive semidefinite. Therefore, the optimization is
convex and can be solved efficiently.




\bibliographystyle{abbrv}
\bibliography{biblio}

\begin{thebibliography}{10}

\bibitem{pc3000}
Pc3000.
\newblock emulab.
\newblock https://wiki.emulab.net/wiki/pc3000.

\bibitem{conga}
M.~Alizadeh, T.~Edsall, S.~Dharmapurikar, R.~Vaidyanathan, K.~Chu,
  A.~Fingerhut, F.~Matus, R.~Pan, N.~Yadav, G.~Varghese, et~al.
\newblock {CONGA}: Distributed congestion-aware load balancing for datacenters.
\newblock In {\em SIGCOMM}, 2014.

\bibitem{apostolopoulos04}
J.~Apostolopoulos and M.~Trott.
\newblock Path diversity for enhanced media streaming.
\newblock {\em IEEE Comm. Mag.}, 42(8), August 2004.

\bibitem{arzani12}
B.~Arzani, R.~Guerin, and A.~Ribeiro.
\newblock A distributed routing protocol for predictable rates in wireless mesh
  networks.
\newblock In {\em Proc.~IEEE ICNP}, Austin, TX, October 2012.

\bibitem{arzani14}
B.~Arzani, A.~Gurney, S.~Cheng, R.~Guerin, and B.~T. Loo.
\newblock Deconstructing {MPTCP} performance.
\newblock In {\em ICNP}, 2014.

\bibitem{pams14}
B.~Arzani, A.~Gurney, S.~Cheng, R.~Guerin, and B.~T. Loo.
\newblock Impact of path characteristics and scheduling policies on {MPTCP}
  performance.
\newblock In {\em Proc. PAMS}, 2014.

\bibitem{QoE1}
A.~Balachandran, V.~Sekar, A.~Akella, S.~Seshan, I.~Stoica, and H.~Zhang.
\newblock Developing a predictive model of quality of experience for {I}nternet
  video.
\newblock In {\em Proc. ACM SIGCOMM}, 2013.

\bibitem{ball15}
M.~Ball.
\newblock The state and future of {N}etflix v. {HBO} in 2015.
\newblock REDEF Originals, May 2015.
\newblock \url{http://bit.ly/2k7LrfD}.

\bibitem{bertsimas}
D.~Bertsimas and J.~N. Tsitsiklis.
\newblock {\em Introduction to linear optimization}, volume~6.
\newblock Athena Scientific Belmont, MA, 1997.

\bibitem{bonte10}
J.~Bonte.
\newblock Online video: What do media companies \emph{Really} want?
\newblock White paper, Cisco Internet Business Solutions Group, July 2010.

\bibitem{UCG}
Y.~Borghol, S.~Ardon, N.~Carlsson, D.~Eager, and A.~Mahanti.
\newblock The untold story of the clones: content-agnostic factors that impact
  {YouTube} video popularity.
\newblock In {\em Proc. SIGKDD}, 2012.

\bibitem{dummynet}
M.~Carbone and L.~Rizzo.
\newblock Dummynet revisited.
\newblock {\em ACM SIGCOMM Computer Communication Review}, 40(2), March 2010.

\bibitem{cdnsummit13}
Content delivery summit 2013.
\newblock \url{http://www.contentdeliverysummit.com/2013/}.

\bibitem{cdnsummit14}
Content delivery summit 2014.
\newblock \url{http://www.contentdeliverysummit.com/2014/}.

\bibitem{cdnsummit15}
Content delivery summit 2015.
\newblock \url{http://www.contentdeliverysummit.com/2015/}.

\bibitem{cdnsummit16}
Content delivery summit 2016.
\newblock \url{http://www.contentdeliverysummit.com/2016/}.

\bibitem{abr16}
C.~Chen, S.~Inguva, A.~Rankin, and A.~Kokaram.
\newblock A subjective study for the design of multi-resolution {ABR} video
  streams with the {VP9} codec.
\newblock In {\em Proc. Intl. Symp. Electronic Imaging: Human Visual
  Perception}, San Francisco, CA, February 2016.

\bibitem{Chen:2004ih}
J.~Chen, S.~Chan, and V.~Li.
\newblock Multipath routing for video delivery over bandwidth-limited networks.
\newblock {\em IEEE J. Select. Areas Commun}, 22(10), 2004.

\bibitem{chen09}
X.~Chen, M.~Chamania, A.~Jukan., A.~Drummond, and N.~D. Fonseca.
\newblock On the benefits of multipath routing for distributed data-intensive
  applications with high bandwidth requirements and multidomain reach.
\newblock In {\em Proc. Comm. Netw. \& Serv. Research Conf. (CNSR'09)},
  Moncton, NB, May 2009.

\bibitem{youtuber}
Y.-C. Chen, D.~Towsley, and R.~Khalili.
\newblock {MSPlayer}: Multi-source and multi-path leveraged {YoutubER}.
\newblock In {\em CoNEXT}, 2014.

\bibitem{longtail}
X.~Cheng, J.~Liu, and C.~Dale.
\newblock Understanding the characteristics of {I}nternet short video sharing:
  A {YouTube}-based measurement study.
\newblock {\em IEEE Transactions on Multimedia}, 2013.

\bibitem{vxr14}
Conviva.
\newblock 2014 viewer experience report 2014.
\newblock
  \url{http://lp.conviva.com/rs/901-ZND-194/images/2014%20Conviva%20Viewer%20Experience%20Report.pdf/},.

\bibitem{conviva15}
Conviva.
\newblock 2015 end of year report and 2016 predictions.

\bibitem{vxr15}
Conviva.
\newblock 2015 viewer experience report.
\newblock
  \url{http://lp.conviva.com/rs/901-ZND-194/images/Conviva_Viewer_Experience_Report_2015_Final.pdf}.

\bibitem{mmsys}
X.~Corbillon, R.~Aparicio-Pardo, N.~Kuhn, G.~Texier, and G.~Simon.
\newblock Cross-layer scheduler for video streaming over {MPTCP}.
\newblock In {\em Proceedings of the 7th International Conference on Multimedia
  Systems}. ACM, 2016.

\bibitem{QoE3}
F.~Dobrian, V.~Sekar, A.~Awan, I.~Stoica, D.~Joseph, A.~Ganjam, J.~Zhan, and
  H.~Zhang.
\newblock Understanding the impact of video quality on user engagement.
\newblock In {\em ACM SIGCOMM Computer Communication Review}, volume~41, pages
  362--373. ACM, 2011.

\bibitem{abr96}
J.~G.~D.~Forney, L.~Brown, M.~V. Eyuboglu, and J.~L.~M. III.
\newblock The {V.34} high-speed modem standard.
\newblock {\em IEEE Comm. Mag.}, 54(12):28--33, December 1996.

\bibitem{ganesan2001highly}
D.~Ganesan, R.~Govindan, S.~Shenker, and D.~Estrin.
\newblock Highly-resilient, energy-efficient multipath routing in wireless
  sensor networks.
\newblock {\em ACM SIGMOBILE Mobile Computing and Communications Review}, 5(4),
  2001.

\bibitem{garcia14}
M.-N. Garcia, F.~D. Simone, S.~Tavakoli, N.~Staelens, S.~Egger,
  K.~Brunnstr\'{o}m, and A.~Raake.
\newblock Quality of experience and {HTTP} adaptive streaming: A review of
  subjective studies.
\newblock In {\em Proc. 6th IEEE Intl. Workshop on Quality of Multimedia
  Experience (QoMEX)}, Singapore, Singapore, September 2014.

\bibitem{gill2007youtube}
P.~Gill, M.~Arlitt, Z.~Li, and A.~Mahanti.
\newblock Youtube traffic characterization: a view from the edge.
\newblock In {\em Proceedings of the 7th ACM SIGCOMM conference on Internet
  measurement}, pages 15--28. ACM, 2007.

\bibitem{golubchik02}
L.~Golubchik, J.~Lui, T.~Tung, A.~Chow, W.-J. Lee, G.~Franceschinis, and
  C.~Anglano.
\newblock Multi-path continuous media streaming: what are the benefits?
\newblock {\em Performance Evaluation}, 49(1--4), September 2002.

\bibitem{QoE2}
T.~Ho{\ss}feld, R.~Schatz, E.~Biersack, and L.~Plissonneau.
\newblock Internet video delivery in {YouTube}: From traffic measurements to
  quality of experience.
\newblock In {\em Data Traffic Monitoring and Analysis}. Springer, 2013.

\bibitem{ishai}
V.~Jalaparti, I.~Bliznets, S.~Kandula, B.~Lucier, and I.~Menache.
\newblock Dynamic pricing and traffic engineering for timely inter-datacenter
  transfers.
\newblock 2016.

\bibitem{javed09}
U.~Javed, M.~Suchara, J.~he, and J.~Rexford.
\newblock Multipath protocol for delay-sensitive traffic.
\newblock In {\em Proc. COMSNETS'09}, Bangalore, India, January 2009.

\bibitem{krishnan12}
S.~S. Krishnan and R.~K. Sitaraman.
\newblock Video stream quality impacts viewer behavior: Inferring causality
  using quasi-experimental designs.
\newblock In {\em Proc. ACM IMC'12}, Boston, MA, November 2012.

\bibitem{kufa15}
J.~Kufa and T.~Kratochvil.
\newblock Comparison of {H.265} and {VP9} coding efficiency for full {HDTV} and
  ultra {HDTV} applications.
\newblock In {\em Radioelektronika (RADIOELEKTRONIKA), 2015 25th International
  Conference}, pages 168--171, April 2015.

\bibitem{multihoming2}
H.~H. Liu, Y.~Wang, Y.~R. Yang, H.~Wang, and C.~Tian.
\newblock Optimizing cost and performance for content multihoming.
\newblock In {\em Proc. ACM SIGCOMM}, 2012.

\bibitem{mosek}
A.~Mosek.
\newblock The {MOSEK} optimization software.
\newblock {\em Online at http://www. mosek. com}, 54, 2010.

\bibitem{broker}
M.~K. Mukerjee, I.~N. Bozkurt, B.~Maggs, S.~Seshan, and H.~Zhang.
\newblock The impact of brokers on the future of content delivery.
\newblock In {\em Proceedings of the 15th ACM Workshop on Hot Topics in
  Networks}, pages 127--133. ACM, 2016.

\bibitem{centralized}
M.~K. Mukerjee, D.~Naylor, J.~Jiang, D.~Han, S.~Seshan, and H.~Zhang.
\newblock Practical, real-time centralized control for {CDN}-based live video
  delivery.
\newblock In {\em Proceedings of the 2015 ACM Conference on Special Interest
  Group on Data Communication}. ACM, 2015.

\bibitem{donar2}
S.~Narayana, W.~Jiang, J.~Rexford, and M.~Chiang.
\newblock Joint server selection and routing for geo-replicated services.
\newblock In {\em Proc. UCC}, 2013.

\bibitem{openconnect}
Netflix, openconnect.
\newblock \url{https://openconnect.itp.netflix.com/}.

\bibitem{ozer16}
J.~Ozer.
\newblock Netflix finds {x265} 20\% more efficient than {VP9}.
\newblock streaming media, September 2016.
\newblock Available at \url{http://bit.ly/2kyVIP0}.

\bibitem{mptcpscheduler14}
C.~Paasch, S.~Ferlin, O.~Alay, and O.~Bonaventure.
\newblock Experimental evaluation of multipath {TCP} schedulers.
\newblock In {\em Proceedings of the 2014 ACM SIGCOMM workshop on Capacity
  sharing workshop}. ACM, 2014.

\bibitem{analdesign}
Q.~Peng, A.~Walid, and S.~H. Low.
\newblock Multipath {TCP}: Analysis and design.
\newblock In {\em Proc. ACM SIGMETRICS}, Pittsburgh, PA, June 2013.

\bibitem{qwiltreport}
6 key factors to consider when choosing a transparent caching solution.
\newblock
  \url{http://qwilt.com/downloads/Qwilt-TransparentCaching-6KeyFactors.pdf}.

\bibitem{Radi:2012en}
M.~Radi, B.~Dezfouli, K.~A. Bakar, and M.~Lee.
\newblock Multipath routing in wireless sensor networks: Survey and research
  challenges.
\newblock {\em Sensors}, 12(1), January 2012.

\bibitem{nsdi2012}
C.~Raiciu, C.~Paasch, S.~Barre, A.~Ford, M.~Honda, F.~Duchene, O.~Bonaventure,
  and M.~Handley.
\newblock How hard can it be? designing and implementing a deployable
  {Multipath TCP}.
\newblock In {\em Proc. NSDI}, 2012.

\bibitem{rao11}
A.~Rao, A.~Legout, Y.-s. Lim, D.~Towsley, C.~Barakat, and W.~Dabbous.
\newblock Network characteristics of video streaming traffic.
\newblock In {\em Proc. ACM CoNEXT}, 2011.

\bibitem{dynamic}
P.~Rodriguez and E.~W. Biersack.
\newblock Dynamic parallel access to replicated content in the {I}nternet.
\newblock {\em IEEE/ACM Transactions on Networking}, 2002.

\bibitem{netflixdeal}
K.~Russel.
\newblock What the {N}etflix-{C}omcast deal really means in plain english.
\newblock
  \url{http://www.businessinsider.com/netflix-comcast-deal-explained-2014-2}.

\bibitem{sandvine16}
Sandvine.
\newblock {\em 2016 Global Internet Phenomena -- Latin America \& North
  America}, 2016.

\bibitem{seufert14}
M.~Seufert, S.~Egger, M.~Slanina, T.~Zinner, and T.~Hossfeld.
\newblock A survey on qualitty of experience of {HTTP} adaptive streaming.
\newblock {\em IEEE Communications Surveys \& Tutorials}, 17(1), April 2014.

\bibitem{conj1}
A.~Singh, M.~Xiang, A.~K{\"o}nsgen, C.~Goerg, and Y.~Zaki.
\newblock Enhancing fairness and congestion control in {M}ultipath {TCP}.
\newblock In {\em Proc. IEEE WMNC}, 2013.

\bibitem{MRTP}
V.~Singh, T.~Karkkainen, J.~Ott, S.~Ahsan, and L.~Eggert.
\newblock Multipath rtp (mprtp).
\newblock Technical report, IETF Internet-Draft, 2012.

\bibitem{sogaard16}
J.~S{\o}gaard, S.~Tavakoli, K.~Brunnstr\"{o}m, and N.~Garcia.
\newblock Subjective analysis and objective characterization of adaptive
  bitrate videos.
\newblock In {\em Proc. Intl. Symp. Electronic Imaging: Human Vision and
  Electroninc Imaging}, San Francisco, CA, February 2016.

\bibitem{shapley}
R.~Stanojevic, N.~Laoutaris, and P.~Rodriguez.
\newblock On economic heavy hitters: shapley value analysis of 95th-percentile
  pricing.
\newblock In {\em SIGCOMM IMC}, 2010.

\bibitem{uhrina14}
M.~Uhrina, L.~Sevcik, J.~Frnda, and M.~Vaculik.
\newblock Impact of {H.265} and {VP9} compression standards on the video
  quality for 4k resolution.
\newblock In {\em Telecommunications Forum Telfor (TELFOR), 2014 22nd}, pages
  905--908, November 2014.

\bibitem{emulab}
K.~Webb, M.~Hibler, R.~Ricci, A.~Clements, and J.~Lepreau.
\newblock Implementing the {Emulab-PlanetLab} portal: Experience and lessons
  learned.
\newblock In {\em Proc. WORLDS}, 2004.

\bibitem{donar}
P.~Wendell, J.~W. Jiang, M.~J. Freedman, and J.~Rexford.
\newblock Donar: decentralized server selection for cloud services.
\newblock In {\em Proc. ACM SIGCOMM}, 2010.

\bibitem{conj2}
D.~Wischik, C.~Raiciu, A.~Greenhalgh, and M.~Handley.
\newblock Design, implementation and evaluation of congestion control for
  {Multipath TCP}.
\newblock In {\em Proc. NSDI}, 2011.

\bibitem{cdnp2p}
H.~Yin, X.~Liu, T.~Zhan, V.~Sekar, F.~Qiu, C.~Lin, H.~Zhang, and B.~Li.
\newblock Design and deployment of a hybrid {CDN-P2P} system for live video
  streaming: experiences with {LiveSky}.
\newblock In {\em ACMMM}, 2009.

\bibitem{netsession}
M.~Zhao, P.~Aditya, A.~Chen, Y.~Lin, A.~Haeberlen, P.~Druschel, B.~Maggs,
  B.~Wishon, and M.~Ponec.
\newblock Peer-assisted content distribution in {Akamai Netsession}.
\newblock In {\em Proc. IMC}, 2013.

\end{thebibliography}
\end{document}